\newif\ifarxiv
\arxivtrue

\documentclass[a4paper,USenglish,cleveref,pagebackref,thm-restate]{lipics-v2021}

\ifarxiv{}
  \nolinenumbers
  \hideLIPIcs
\fi{}
\usepackage{amsmath}
\usepackage{amsfonts}
\usepackage{nicefrac}
\usepackage{amssymb}
\usepackage{amsthm}
\usepackage{mathtools}
\usepackage{tikz}
\usetikzlibrary{calc}

\usepackage[backgroundcolor=gray!10,textsize=footnotesize]{todonotes}

\usepackage[square,numbers,sort&compress]{natbib}
\usepackage{doi}

\theoremstyle{definition}
\newtheorem{assumption}[theorem]{Assumption}
\newtheorem{algorithm}[theorem]{Algorithm}

\crefname{lemma}{Lemma}{Lemmata}

\newcommand{\problemdef}[3]{
	\begin{center}
	\begin{minipage}{0.95\textwidth}
		\noindent
		#1
		\vspace{5pt}\\
		\setlength{\tabcolsep}{3pt}
		\begin{tabularx}{\textwidth}{@{}lX@{}}
			\textbf{Input:}     & #2 \\
			\textbf{Question:}  & #3
		\end{tabularx}
	\end{minipage}
	\end{center}
}
\newcommand{\optproblemdef}[3]{
	\begin{center}
	\begin{minipage}{0.95\textwidth}
		\noindent
		#1
		\vspace{5pt}\\
		\setlength{\tabcolsep}{3pt}
		\begin{tabularx}{\textwidth}{@{}lX@{}}
			\textbf{Input:}     & #2 \\
			\textbf{Output:}  & #3
		\end{tabularx}
	\end{minipage}
	\end{center}
}

\DeclarePairedDelimiterX{\abs}[1]{\lvert}{\rvert}{#1}
\DeclarePairedDelimiterX{\norm}[1]{\lVert}{\rVert}{#1}
\DeclarePairedDelimiterX{\ceil}[1]{\lceil}{\rceil}{#1}

\newcommand{\NN}{\mathbb{N}}
\newcommand{\ZZ}{\mathbb{Z}}

\newcommand{\JJ}{\mathbb{J}}
\newcommand{\BB}{\mathcal{B}}
\newcommand{\UU}{\mathcal{U}}
\newcommand{\TT}{\mathcal{T}}
\newcommand{\II}{\mathcal{I}}

\newcommand{\lifetime}{\tau}

\newcommand{\NP}{\textrm{NP}}

\newcommand{\Wone}{\textrm{W[1]}}
\newcommand{\FPT}{\textrm{FPT}}

\newcommand{\bigO}{\mathcal{O}}
\newcommand{\yes}{\emph{yes}}
\newcommand{\no}{\emph{no}}

\DeclareMathOperator{\sgn}{sgn}
\DeclareMathOperator{\pf}{pf}
\DeclareMathOperator{\val}{val}
\DeclareMathOperator{\tw}{tw}

\newcommand{\oneto}[1]{[ #1 ]} %

\newcommand{\wilog}{without loss of generality}
\newcommand{\Wilog}{Without loss of generality}

\usepackage{etoolbox}
\newcommand{\appsymb}{$\bigstar$}
\newcommand{\appref}[1]{\ifarxiv{}\else{}\hyperref[proof:#1]{\appsymb}\fi{}}

\newcommand{\appendixsection}[1]{%
		\ifarxiv{}\else{}%
	\gappto{\appendixProofText}{\section{Additional Material for Section~\ref{#1}}\label{app:#1}}%
	\fi{}%
}

\newcommand{\toappendix}[1]{%
\ifarxiv{}#1\else{}%
  \gappto{\appendixProofText}
  {{
    #1
  }}
  \fi{}%
}

\newcommand{\appendixproof}[2]{%
		\ifarxiv{}%
		#2
\else{}%
  \gappto{\appendixProofText}
  {
    \subsection[Missing Proof]{Proof of \cref{#1}}\label{proof:#1}
    #2
  }%
  \fi{}%
}

\title{Parameterized Algorithms for Diverse Multistage Problems}

\titlerunning{Parameterized Algorithms for Diverse Multistage Problems}

\author{Leon Kellerhals}{Algorithmics and Computational Complexity, Faculty IV, Technische Universität Berlin, Germany}{leon.kellerhals@tu-berlin.de}{https://orcid.org/0000-0001-6565-3983}{}

\author{Malte Renken}{Algorithmics and Computational Complexity, Faculty IV, Technische Universität Berlin, Germany}{m.renken@tu-berlin.de}{https://orcid.org/0000-0002-1450-1901}{Supported by the German Research Foundation (DFG), project MATE (NI~369/17).}

\author{Philipp Zschoche}{Algorithmics and Computational Complexity, Faculty IV, Technische Universität Berlin, Germany}{zschoche@tu-berlin.de}{https://orcid.org/0000-0001-9846-0600}{}

\acknowledgements{
The authors wish to thank Rolf Niedermeier for his useful comments on the manuscript.
This work was initiated at the research retreat of the Algorithmics and Computational Complexity group of TU Berlin in September 2020 in Zinnowitz.
}

\authorrunning{L.~Kellerhals, M.~Renken, and P.~Zschoche}

\Copyright{Leon~Kellerhals, Malte~Renken, and Philipp Zschoche}



\begin{document}

\maketitle
\begin{abstract}
	The world is rarely static --- many problems need not only be solved once but repeatedly, under changing conditions.
	This setting is addressed by the ``multistage'' view on computational problems.
	We study the ``diverse multistage'' variant, where consecutive solutions of large variety are preferable to similar ones,
	e.g.\ for reasons of fairness or wear minimization.
	While some aspects of this model have been tackled before,
	we introduce a framework allowing us to prove that a number of diverse multistage problems
	are fixed-parameter tractable by diversity,
	namely \textsc{Perfect Matching}, \textsc{$s$-$t$ Path}, \textsc{Matroid Independent Set}, and \textsc{Plurality Voting}.
	This is achieved by first solving special, colored variants of these problems, which might also be of independent interest.
\end{abstract}

\ifarxiv{}\else{}
\clearpage
\fi{}

\section{Introduction}

In the \emph{multistage} setting, given a sequence of instances of some problem,
one asks whether there is a corresponding sequence of solutions such that consecutive solutions relate in some way to each other.
Often the aim is to find consecutive solutions that are very \emph{similar}~\cite{gupta2014changing,eisenstat2014facility,FluschnikNRZ19,BampisET19,BampisEST19,F21}.
This is reasonable when changing between distinct solutions incurs some form of cost.
In other settings, the opposite goal is more reasonable, that is, consecutive solutions should be very \emph{different}.
This is a natural goal when wear minimization, load distribution, or resilience against failures or attacks are of interest.
This ``diverse multistage'' setting is what we want to focus on in this paper.
Here, given a sequence of instances of some decision problem,
the task is to find a sequence of solutions such that the \emph{diversity}, i.e., the size of the symmetric difference 
of any two consecutive solutions is \emph{at least}~$\ell$.

This problem has already received some attention in the literature:
\citet{FNSZ20} studied the problem of finding diverse $s$-$t$ paths
and
\citet{bredereck2020multistage} considered series of committee elections.
In a similar setting, but aiming for large symmetric difference between every two (i.e., not just consecutive) solutions,
\citet{BasteDiversity20} provide a framework for parameterization by treewidth,
while 
\citet{fomin2020diverse,FominGPP021} focus on the case that all problems are defined on the same graph and study matching, independent set, and matroids.

\newcommand{\solutionset}[1][I]{{\mathcal R({#1})}}
\newcommand{\dmPi}{\textsc{Diverse Multistage $\Pi$}}

We briefly give a formal definition.
Assume $\Pi$~to be some decision problem which asks 
whether the \emph{family of solutions} $\solutionset \subseteq 2^{\BB(I)}$ of an instance~$I$ of~$\Pi$ is non-empty, where $\BB(I)$ 
is some \emph{base set} encompassing all possible solutions.
For example, for an instance~$I$ of~$\textsc{Vertex Cover}$, the set $\BB(I)$ is the set of all vertices
and $\solutionset$~is the set of all vertex covers within the size bound.
The problem $\dmPi$ is now the following.
\problemdef{\dmPi}
{A sequence $(I_i)_{i=1}^\lifetime$ of instances of $\Pi$ and an integer $\ell \in \NN_0$.}
{Is there a sequence $(S_i)_{i=1}^\lifetime$ of solutions~$S_i \in \solutionset[I_i]{}$ 
such that $|S_{i}\Delta S_{i+1}| \geq \ell$ for all~$i \in \oneto{\tau-1}$?}

\subparagraph{Our contributions.}
We present a general framework which allows us to prove fixed-parameter tractability of \dmPi{} parameterized by the diversity~$\ell$ for several problems~$\Pi$.
This includes finding diverse matchings,
but also diverse commitees (answering an open question by \citet{bredereck2020multistage}),
diverse $s$-$t$ paths, and diverse independent sets in matroids such as spanning forests.
Finally, we show that similar results cannot be expected for finding diverse vertex covers.

Generally, our framework can be applied to \dmPi{} whenever one can solve a $4$-colored variant of $\Pi$ efficiently.
Formally, this variant is defined as follows.
\newcommand{\coloredPi}{\textsc{4-Colored Exact $\Pi$}}
\optproblemdef{\coloredPi}
{An instance $I$ of $\Pi$, a coloring $c \colon \BB(I) \rightarrow [4]$, and $n_i \in \NN_0,i\in[4]$.}
{A solution $S \in \solutionset$ such that $\abs{\{x \in S \mid c(x) = i\}} = n_i$ for all~$i \in[4]$
or ``\no{}'' if no solution exists.}
Our main result reads as follows.
\begin{theorem}
	\label{thm:diverse-multistage-pi-det}
	If an instance $I$ of \coloredPi{} can be solved in $f(r)\cdot |I|^{O(1)}$ time,
	then an instance $J$ of \dmPi{} of size $n$ can be solved in $2^{O(\ell)} \cdot f(r_{\max}) \cdot \abs{J}^{O(1)}$ time,
	where $r_{\max}$ is the maximum of parameter $r$ over all instances of $\Pi$ in $J$.\footnote{For example, if the input is a sequence of graphs and $r$ is the treewidth, then $r_{\max}$ is the maximum treewidth over all graphs in the input.}
\end{theorem}
We prove \cref{thm:diverse-multistage-pi-det} in \cref{sec:framework} in a more general form which also allows solving \coloredPi{} by a Monte Carlo algorithm.
We then apply our framework to the following problems:

\newcommand{\Voting}{\textsc{Plurality Voting}}
\newcommand{\dmVoting}{\textsc{Diverse Multistage \Voting}}
\newcommand{\coloredVoting}{\textsc{$4$-Colored Exact \Voting}}

\emph{Committee Election (\cref{sec:committee}).}
In \dmVoting{}, we are given a set $A$ of agents, a set $C$ of candidates,
and $\lifetime$ many \emph{voting profiles}~$u_i \colon A \to C$.
The goal is to find a sequence $(C_i)_{i=1}^\lifetime$ of committees~$C_i \subseteq C$ 
such that each committee~$C_i$ is of size at most~$k$ and gets at least $x$~votes in the voting profile~$u_i$
(i.e., $\abs{u_i^{-1}(C_i)} \geq x$),
and $\abs{C_{i} \Delta C_{i+1}} \geq \ell$ for all~$i \in \oneto{\tau-1}$.
We show that there is a $2^{\bigO(\ell)}\cdot \abs{J}^{\bigO(1)}$-time algorithm to solve a \dmVoting{} instance $J$.
This answers an open question of \citet{bredereck2020multistage}.
Later, in \cref{sec:matroids}, we generalize the algorithm 
used to solve \coloredVoting{} to matroids.

\newcommand{\dmMatching}{\textsc{Diverse Multistage Perfect Matching}}
\newcommand{\ExactMatchingProb}{\textsc{$s$-Colored Exact Perfect Matching}}
\emph{Perfect Matching (\cref{sec:matching}).}
In the multistage setting, Perfect Matching is among the problems most intensively studied~\cite{gupta2014changing,bampis2018multistage,bampis2019lp,chimani2020approximating,steinhau2020parameterized}.
Given a sequence of graphs $(G_i)_{i=1}^\lifetime$ and an integer $\ell$,
\dmMatching{} asks whether there is a sequence~$(M_i)_{i=1}^\lifetime$ such that
each~$M_i$ is a perfect matching in~$G_i$, 
and $\abs{M_{i} \Delta M_{i+1}} \geq \ell$ for all~$i \in \oneto{\tau -1}$.
We show that there is a randomized $2^{\bigO(\ell)}\cdot\abs{J}^{\bigO(1)}$-time algorithm to solve a \dmMatching{} instance $J$ with constant %
error probability.
This stands in remarkable contrast to the $\Wone$-hardness of the 
(non-diverse) \textsc{Multistage Perfect Matching}, 
when parameterized by $\ell+\lifetime$~\cite{steinhau2020parameterized}.
To apply our framework, we establish an algebraic algorithm using the Pfaffian of a specific variant of the Tutte matrix to solve \ExactMatchingProb{} on an $n$-vertex graph
in $n^{O(s)}$ time with low %
error probability.

\newcommand{\stPath}{\textsc{$s$-$t$ Path}}
\newcommand{\dmPath}{\textsc{Diverse Multistage \stPath}}
\newcommand{\coloredPath}{\textsc{$4$-Colored Exact \stPath}}
\emph{$s$-$t$ Path (\cref{sec:path}).}
Studying \stPath{} in the multistage setting was already suggested in the seminal work of \citet{gupta2014changing}.
In \dmPath{} one is given a sequence of graphs $(G_i)_{i=1}^\lifetime$,
two distinct vertices $s$ and $t$, and an integer $\ell$, and 
asks whether there is a sequence $(P_i)_{i=1}^\lifetime$ such that
each~$P_i$ is an \stPath{} in $G_i$,
and $\abs{V(P_{i}) \Delta V(P_{i+1})} \geq \ell$ for all~$i \in \oneto{\tau -1}$.
\citet{FNSZ20} provided a comprehensive study of finding $s$-$t$ paths of bounded length
in the multistage setting from the viewpoint of parameterized complexity.
Among other results, they showed that \dmPath{} is $\NP$-hard
but fixed-parameter tractable when parameterized
by the maximum length of an \stPath{} in the solution.
We show that \dmPath{} parameterized by~$\ell$ is fixed-parameter tractable.
At first glance, using our framework seems unpromising since \coloredPath{} can presumably not be solved in polynomial time 
(it is NP-hard by a straight-forward reduction from \textsc{Hamiltonian Path}).
However, we develop a win/win strategy around a generalization of the Erd\H{o}s-P\'{o}sa theorem for long cycles due to \citet{MNSW-long-cycles}
so that we have to solve \coloredPath{} only on graphs on which the treewidth is upper-bounded in the parameter $\ell$.

In~\cref{sec:vertex-cover}, we complement our fixed-parameter tractability results with a \Wone-hardness for \textsc{Diverse Multistage Vertex Cover} 
when parameterized by~$\ell$.

\section{Preliminaries}
\label{sec:prem}
We denote by~$\NN$ and~$\NN_0$ the natural numbers excluding and including zero, respectively.
For~$n\in\NN$, let~$\oneto{n} \coloneqq \{1, 2, \dots, n\}$.
For two sets~$A$ and~$B$,
we denote by~$A \Delta B \coloneqq (A\setminus B) \cup (B \setminus A)=(A\cup B)\setminus (A\cap B)$ the \emph{symmetric difference} of~$A$ and~$B$,
and by~$A\uplus B$ the disjoint union of~$A$ and~$B$.
For a function $c \colon A \to B$, let~$c(A') \coloneqq \bigcup_{a \in A'} c(a)$ 
and~$c^{-1}(b) \coloneqq \{ a \in A \mid c(a) = b \}$, where $A' \subseteq A$.
We also use the notations~$c^b$ and~$c^{b,b'}$ as shorthands for~$c^{-1}(b)$ and~$c^{-1}(b) \cup c^{-1}(b')$, respectively.

A \emph{Monte Carlo algorithm}, or an algorithm with error probability $p$, 
is a randomized algorithm that returns a correct answer with probability $1-p$.

Let~$\Sigma$ be a finite alphabet.
A \emph{parameterized problem}~$L$ is a subset~$L\subseteq \{(x,k)\in\Sigma^*\times \NN_0\}$.
An instance~$(x,k)\in\Sigma^*\times \NN_0$ is a \yes-instance of~$L$ if and only if~$(x,k)\in L$ (otherwise, it is a \no-instance).
A parameterized problem~$L$ is \emph{fixed-parameter tractable} (in \FPT) 
if for every input~$(x,k)$ one can decide in~$f(k)\cdot |x|^{O(1)}$~time whether~$(x,k)\in L$, 
where~$f$ is some computable function only depending on~$k$.  
A \Wone-hard parameterized problem is not fixed-parameter tractable unless~\FPT=\Wone.
\ifarxiv{}We refer to \citet{downey2013fundamentals} and \citet{cygan2015parameterized} for more material on parameterized complexity.\fi{}We use standard notation from graph theory~\cite{Diestel10}.
Throughout this paper, we assume graphs to be simple and undirected.

\section{The General Framework}
\label{sec:framework}
\appendixsection{sec:framework}
In this section, we introduce a general framework to show (for some decision problem $\Pi$) fixed-parameter tractability of \dmPi{} parameterized by~$\ell$.
Recall that, for every instance $I$ of decision problem $\Pi$, we denote the family of solutions by $\solutionset \subseteq 2^{\BB(I)}$ and 
the input size $|I|$ of $I$ is at least $|\BB(I)|$. For the reminder of this section we assume that~$|\BB(I)| \geq 2$ for all instances $I$ of $\Pi$.
The framework is applicable to \dmPi{} if there is an efficient algorithm for \coloredPi{}.
Formally, we use the following prerequisite, which is slightly more general than in \cref{thm:diverse-multistage-pi-det}.
\newcommand{\Pialgo}{\ensuremath{\mathcal{A}}}
\newcommand{\param}{r}
\begin{assumption}
		\label{assumption}
		There are computable functions $f, g$ such that
		for every $0 \leq p \leq 1$
		for which $g(p)$ is defined,
		there is a Monte-Carlo algorithm \Pialgo{} with error probability~$p$ and 
		running time~$f(\param) \cdot |I|^{O(1)} \cdot g(p)$, that solves an instance~$I$ of \coloredPi{},
		where $\param\in \NN_0$ is some parameter of $I$ and $g$ is monotone non-increasing.
\end{assumption}
We allow an error probability in \cref{assumption} 
because for one of our applications (in \cref{sec:matching}),
no other polynomial-time algorithm is known.
The goal %
is to prove the following.
\ifarxiv{}%
\begin{theorem}
\else{}%
\begin{theorem}[\appref{thm:diverse-multistage-pi}]
\fi{}
		\label{thm:diverse-multistage-pi}
		Let \cref{assumption} be true.
		Then any size-$n$ instance~$I$ of \dmPi{} can be solved
		in $2^{O(\ell)}\cdot f(\param_{\max})\cdot n^{O(1)} \cdot g(\nicefrac{p}{\lifetime 2^{O(\ell)}n^{O(1)}})$~time
		by a Monte-Carlo algorithm with error probability~$p$,
		where $r_{\max}$ is the maximum of parameter $r$ over all instances of $\Pi$ in $I$,
		and $0 \leq p \leq 1$
		is an arbitrary probability for which the above expression is defined.\footnote{For example, if we only have an algorithm with non-zero error probability, then $p=0$ is excluded.}
\end{theorem}
\ifarxiv{}The proof of \cref{thm:diverse-multistage-pi} is deferred to the end of this section.\fi{}
Note that, if we have a non-randomized algorithm in \cref{assumption} (that is, $g(0)$ is defined and $g$ maps always to one),
then \cref{thm:diverse-multistage-pi-det} follows directly from \cref{thm:diverse-multistage-pi}.

The underlying strategy of the algorithm for a \dmPi{}-instance $J$ behind \cref{thm:diverse-multistage-pi} is to compute
for each instance $I$ of $\Pi$ in $J$ a solution family such that the Cartesian product of these families
contains a solution for $J$ if and only if $J$ is a \yes-instance.
Once these families are obtained, we can check whether $J$ is a \yes-instance by dynamic programming.
To this end, we compute a small subset of $\solutionset$ satisfying the following definition.
\newcommand{\F}{\mathcal{F}}
\newcommand{\wF}{\widehat{\mathcal{F}}}
\newcommand{\wS}{\widehat{S}}
\begin{definition}
		\label{def:diverse-families}
		Let $\F$ be a set family.
		A subfamily of $\wF \subseteq \F$ is called an \emph{$\ell$-diverse representative} of $\F$
		if, for any $S\in\F$ and sets~$A, B$ 
		with $\min\{|A \Delta S|,|B \Delta S|\} \geq \ell$,
		there is an $\wS \in \wF$ such that $\min\{|A \Delta \wS|,|B \Delta \wS|\} \geq \ell$.
\end{definition}
\noindent
First of all, we note that $\ell$-diverse representatives can be rather small.
\begin{lemma}
		\label{lem:three-matchings}
		Let $\F$ be a set family and $S_1, S_2, S_3 \in \F$.
		If $|S_i \Delta S_j| \geq 2\ell$ for all distinct $i,j \in \oneto{3}$,
		then $\{ S_1,S_2,S_3\}$ is an $\ell$-diverse representative of $\F$.
\end{lemma}
\begin{proof}
		Assume for contradiction that there exist sets $A$ and $B$
		with $\min\{\abs{A \Delta S_i}, \abs{B \Delta S_i}\} < \ell$ for all~$i$.
		Without loss of generality, assume that $\abs{A \Delta S_1} < \ell$.
		Then for $j \in \{2, 3\}$ we have
		$\abs{A \Delta S_j} \geq \abs{S_1 \Delta S_j} - \abs{S_1 \Delta A} > 2\ell - \ell = \ell$
		by the triangle inequality.
		Therefore, $\abs{B \Delta S_2} < \ell$.
		Again, by the triangle inequality
		$\abs{B \Delta S_3} \geq \abs{S_2 \Delta S_3} - \abs{S_2 \Delta B} > 2\ell - \ell = \ell$,
		i.e., $\min\{\abs{A \Delta S_3}, \abs{B \Delta S_3}\} \geq \ell$ --- a contradiction.
\end{proof}
In the following, we measure the distance of two solutions by the size of the symmetric difference. 
In a nutshell, we compute an $\ell$-diverse representative of the family of solutions by
first trying to compute three solutions which are far apart from each other (that is, 	size of symmetric difference at least $2\ell$). 
If this succeeds, then by \cref{lem:three-matchings} we are done.
Otherwise, we distinguish between three cases. 
\begin{description}
		\item[No solution.] If there is no solution at all, 
				then trivially~$\emptyset$ is an $\ell$-diverse representative of the family of solutions.
			\item[One solution.] If we only find one solution $S_1$ to the instance of $\Pi$,
				then each other solution is close to $S_1$.
				Hence, for any two sets $A,B$, if one of them is far away from~$S_1$, then by the triangle inequality it is also far away from every other solution and can be safely ignored.
				For those sets which are close to~$S_1$, we can exploit the upper bound on the symmetric difference 
				by using color-coding \cite{alon1995color} and then applying \cref{assumption}
				to compute an $\ell$-diverse representative of the family of solutions.
				This case is handled in \cref{lem:one-matchings}.
		\item[Two solutions.] If we find two diverse solutions $S_1$ and $S_2$ such that no other solution is far away from both,
				then $S_1$ and $S_2$ partition the solution space into two parts: 
				the solutions close to $S_1$ and those close to $S_2$.
				Again, given two sets~$A, B$, if either of them is far away from~$S_1$ and $S_2$, then we may ignore it.
				By including $S_1$~and~$S_2$ in our family,
				we may further assume that $A$ is similar to~$S_1$ and $B$ is similar to~$S_2$.
				We distinguish two subcases.
				If the distance between $S_1$~and~$S_2$ is very large,
				then $A$~is far away from all solutions in the second part
				and $B$~is far away from all solutions in the first part.
				We can thus ignore one of them (say $B$) and 
				exploit the fact that $A$, $S_1$, and all solutions of interest are close to each other
				to use color-coding and then apply \cref{assumption}.
				In the other subcase where the distance between $S_1$~and~$S_2$ is bounded,
				we can utilize that fact similarly.
				This case is handled in \cref{lem:two-matchings}.
\end{description}

\noindent
Hereafter, the details.
Before we dive into the case distinction outlined above,
we need to prove two technical lemmata,
telling us how to build a diverse representative set
that works for all sets obeying some given coloring of the elements of $\BB(I)$.
These will later work as building blocks in the construction of proper diverse representatives.
In the first lemma,
only two colors are used,
and we are only concerned with one arbitrary set~$A$ instead of two.
\ifarxiv{}
\begin{lemma}
\else{}
\begin{lemma}
		[\appref{lem:colored-graph-stay-close-to-M}]
\fi{}
		\label{lem:colored-graph-stay-close-to-M}
		Let \cref{assumption} be true.
		Given an instance $I$ of $\Pi$ of size $n$, 
		a coloring~$c \colon \BB(I) \to \oneto{2}$,
		and a solution $M \in \solutionset$,
		one can compute in $f(\param) n^{O(1)} g(pn^{-4})$~time
		and with error probability at most~$p$
		a family~$\F \subseteq \solutionset$ of size at most~$n^4$
		such that for any $S \in \solutionset$
		and any~$A \subseteq \BB(I)$ with
		$S \setminus A \subseteq c^1$ and
		$A \setminus S \subseteq c^2$,
		there is $\wS \in \F$ with 
		$|A\Delta \wS| \geq |A\Delta S|$ and $|M\Delta \wS| = \abs{M \Delta S}$.
\end{lemma}
\appendixproof{lem:colored-graph-stay-close-to-M}{
\begin{proof}
		Let %
		$F'_{1} := c^1 \cap M$, $F'_{2} := c^2 \cap M$, $F'_{3} := c^1 \setminus M$, and $F'_{4} := c^2 \setminus M$.

		Start with $\F = \emptyset$.
		Then, for each $m \leq n$ and each partition $\sum_{i=1}^4 m_i = m$,
		use algorithm~\Pialgo{} to search in $f(\param)n^{O(1)}g(pn^{-4})$~time
		and with error probability at most~$pn^{-4}$
		for a set $N \in \solutionset$ such that $|N\cap F'_i| = m_i$ for all $i \in \oneto{4}$.
		If this succeeds, then we add $N$ to $\F$.
		Since there are $\binom{n+4}{4} \leq n^4$ possibilities for $m_1, \dots, m_4$,
		the probability of an error occurring is upper-bounded by $p$.
		Moreover, the size of $\F$ is upper-bounded by~$n^4$
		and hence the time required is bounded by $f(\param)n^{O(1)}g(pn^{-4})$.
		
		It remains to be proven that $\F$ has the desired properties.
		Let $S \in \solutionset$ be arbitrary		
		and set $m_i := |S \cap F'_i|$ for all $i \in \oneto{4}$.
		By construction, $\F$~contains a set~$\wS \in \solutionset$
		such that $|\wS \cap F'_i| = m_i$.
		We then have $|\wS \Delta M|= m_3 + m_4 + \abs{M} - m_1 - m_2 = \abs{S \Delta M}$.
			
		Let $A \subseteq \BB(I)$ be a set with
		$(S \setminus A) \subseteq c^1$ and
		$(A \setminus S) \subseteq c^2$.
		Since $A \setminus S \subseteq c^2$ we have
		\begin{align}
			\abs{ A \cap S \cap c^1 } = \abs{ A \cap c^1} \geq \abs{ A \cap \wS \cap c^1}
			\label{eq:ASE1}
		\end{align}
		and since $S \setminus A \subseteq c^1$, we have that 
		\begin{align}
			|A\cap S \cap c^2| = |S \cap c^2|= m_2+m_4 = |\wS \cap c^2| \geq |A \cap \wS \cap c^2|.
			\label{eq:ASE2}
		\end{align}
		By adding \eqref{eq:ASE1} and \eqref{eq:ASE2} we obtain 
		$\abs{A \cap S} \geq \abs{A \cap \wS}$
		which in turn implies
		$\abs{A \Delta S} \leq \abs{A \Delta \wS}$
		since $\abs{S} = \abs{\wS}$.
\end{proof}
}%
The next lemma 
extends the approach of \cref{lem:colored-graph-stay-close-to-M}
to the case where we have four colors and two arbitrary sets~$A, B$.
\begin{lemma}
		\label{lem:colored-graph-small-syms}
		Let \cref{assumption} be true.
		Given an instance $I$ of $\Pi$ of size $n$, 
		a coloring~$c \colon \BB(I) \to \oneto{4}$,
		one can compute in $f(\param) n^{O(1)} g(pn^{-4})$~time 
		and with error probability at most~$p$
		a family~$\F \subseteq \solutionset$ of size at most~$n^4$
		such that for any $S \in \solutionset$
		and	all sets $A,B \subseteq \BB(I)$ with
		$
			A \setminus (B\cup S) 
			\subseteq c^1, 
			B \setminus (A\cup S) 
			\subseteq c^2, 
			(A \cap B) \setminus S %
			\subseteq c^3, \text{ and } 
			S \setminus (A \cap B) %
			\subseteq c^4,
		$
		there is $\wS \in \F$ with 
		$|C\Delta \wS| \geq |C \Delta S|$ for all $C \in \{A,B\}$.
\end{lemma}
\begin{proof}
		Begin with $\F = \emptyset$.
		Then, for each $m \leq n$ and each partition $\sum_{i=1}^4 m_i = m$, 
		use algorithm \Pialgo{} to search in $f(\param) n^{O(1)} g(pn^{-4})$~time
		and with error probability at most~$pn^{-4}$
		for an $M \in \solutionset$ such that $|M\cap c^i| = m_i$ for all $i \in \oneto{4}$.
		If this succeeds, then add $M$ to $\F$.
		Since there are $\binom{n+4}{4} \leq n^4$ possibilities for $m_1, \dots, m_4$,
		the probability of an error occurring is upper-bounded by~$p$.
		Moreover, the size of $\F$ is at most~$n^4$
		and thus the overall running time is $f(\param) n^{O(1)} g(pn^{-4})$.
		
		Now let $S \in \solutionset$ be arbitrary.		
		Set $m_i := |S \cap c^i|$, for all $i\in \oneto{4}$.
		By construction there is $\wS \in \F$ such that $|\wS \cap c^i| = m_i$ for all $i$.
		It remains to be proven that $\wS$~has the desired properties.
		To this end, let $A,B \subseteq \BB(I)$ be two sets as stated in the lemma.
		By symmetry, it suffices to show that $\abs{A \Delta \wS} \geq \abs{A \Delta S}$.
		
		Since $S \setminus A \subseteq c^4$ we have
		\begin{align}
		\abs{S \cap A \cap c^{1, 3}} = \abs{S \cap c^{1, 3}} = m_1 + m_3 = \abs{\wS \cap c^{1, 3}} \geq \abs{\wS \cap A \cap c^{1, 3}}
		\label{eq:ASE'}
		\end{align}
		and since $A \setminus S \subseteq c^{1,3}$, we have
		\begin{align}
		\abs{S \cap A \cap c^{2,4}} = \abs{A \cap c^{2,4}} \geq \abs{\wS \cap A \cap c^{2,4}}.
		\label{eq:ASE''}
		\end{align}
		By adding \eqref{eq:ASE'} and \eqref{eq:ASE''}, we obtain
		$\abs{S \cap A} \geq \abs{\wS \cap A}$
		and thus
		$\abs{S \Delta A} \leq \abs{\wS \Delta A}$
		since $\abs{S} = \abs{\wS}$.
\end{proof}
We now describe how we generate the colorings required for using \cref{lem:colored-graph-small-syms,lem:colored-graph-stay-close-to-M}.
Color-coding~\cite{alon1995color} is well-established in the toolbox of parameterized algorithms.
While color-coding was initially described as a randomized technique, we use universal sets \cite{NearOptimalDerandomization}
to derandomize this technique as shown in the next lemma.
Interestingly, without this derandomization the error probability of the color-coding step would later propagate 
through the dynamic program and consequently also depend on the number of instances of $\Pi$ in the input instance of \dmPi{}.
The derandomization works as follows.

\newcommand{\famsize}{2^{2b+o(b)} \log n}
\ifarxiv{}
\begin{lemma}
\else{}
\begin{lemma}
		[\appref{thm:derandomized-color-coding}]
\fi{}
		\label{thm:derandomized-color-coding}
For any set~$A$ of size~$n$ and any $b \leq n$ one can compute in $\famsize{} \cdot n$ time a family
of functions $\{c_j \colon A \to \oneto{4} \mid j \in \oneto{\famsize}\}$
such that for any $\biguplus_{i=1}^4 B_i \subseteq A$ with $\abs{\biguplus_{i=1}^4 B_i} \leq b$
there is a~$j$
such that $c_j(B_i) = \{i\}$, for all $i\in [4]$.
\end{lemma}
\appendixproof{thm:derandomized-color-coding}{
\begin{proof}
Let $A := \{a_1, \dots, a_{n}\}$.
By a result of
\citet{NearOptimalDerandomization},
one can compute in $2^{2b}b^{\bigO(\log b)} \cdot \log n \cdot n \subseteq \famsize{} \cdot n$ time a so-called \emph{$(2n, 2b)$-universal set}
which is a family~$\UU \subseteq 2^{\oneto{2n}}$ such that for every $B' \subseteq A$ with~$|B'| = 2b$
the family $\{B' \cap U \mid U \in \UU\}$ contains all $2^{2b}$ subsets of $B'$.
Let~$\UU := \{U_i\}_{i=1}^{\famsize}$.
We then define $c_j$, $j \in \famsize$, by
\[
	c_j(a_i) := 
	\begin{cases}
		1, & \text{if } i, i + n \in U_j, \\
		2, & \text{if } i \in U_j \text{ and } i + n \notin U_j, \\
		3, & \text{if } i \notin U_j \text{ and } i + n \in U_j, \text{ and} \\
		4, & \text{if } i, i + n \notin U_j.
	\end{cases}
\]
Now let $B_1 \uplus B_2 \uplus B_3 \uplus B_4 \subseteq A$ be an arbitrary $4$-partition of a subset of $A$ of size at most~$b$.
Consider $B' := \{i, i+ n \mid a_i \in \bigcup_{q=1}^4 B_q\}$.
We assume that $B'$ is of size $b$, otherwise we add arbitrary elements from $[2n]$.
Since $B'':=\{i, i+ n \mid a_i \in B_1\} \cup \{i \mid a_i \in B_2\} \cup \{i+n \mid a_i \in B_3\} \subseteq B'$
there is an~$U_j \in \UU$ such that $B' \cap U_j = B''$.
Hence, $c_j(B_i) = \{i\}$, for all $i\in [4]$.
\end{proof}
}%

\noindent We now show how to generate an $\ell$-diverse representative of the family of solutions
if there is one solution~$M^*$ from which no other solution differs by more than~$2\ell$.
\newcommand{\famsizetwo}{2^{16\ell+o(\ell)} \log n}
\ifarxiv{}
\begin{lemma}
\else{}
\begin{lemma}
		[\appref{lem:one-matchings}]
\fi{}
		\label{lem:one-matchings}
		Let \cref{assumption} be true.
		Given an instance $I$ of $\Pi$ of size $n$, 
		and a solution $M^* \in \solutionset$ 
		such that each $M \in \solutionset$ satisfies $|M \Delta M^*| \leq 2\ell$,
		one can compute in $\famsizetwo \cdot f(\param)\cdot n^{O(1)} \cdot g(\nicefrac{p}{\famsizetwo\cdot n^4})$ time
		and with error probability~$p$
		an $\ell$-diverse representative of $\solutionset$ of size at most $\famsizetwo \cdot n^4$.
\end{lemma}
\appendixproof{lem:one-matchings}{
\begin{proof}
		For simplicity, let $\JJ := [\famsizetwo]$.
		Apply \cref{thm:derandomized-color-coding} with $b = 8\ell$ to compute
		in $\famsizetwo \cdot n$~time
		a family of colorings $\{c_j \colon \BB(I) \to \oneto{4} \mid j \in \JJ\}$.
		By \cref{thm:derandomized-color-coding} this family has size $\abs{\JJ}$.
		For each $j\in \JJ$, apply \cref{lem:colored-graph-small-syms} to $I$ and $c_j$
		to compute a family~$\F_j \subseteq \solutionset$
		with error probability~$p\cdot\abs{\JJ}^{-1}$.
		Observe that the probability of an error occurring at any of the~$\abs{\JJ}$ steps is bounded by $p$.
		Choose $\F := \{M^*\} \cup \bigcup_{j \in \JJ} \F_j$.
		According to \cref{lem:colored-graph-small-syms}
		the size of $\F$ is upper-bounded by $\abs{\JJ} \cdot n^4$
		and the time required is bounded by $\abs{\JJ} \cdot f(\param)\cdot n^{O(1)}\cdot g(p n^{-4}\cdot \abs{\JJ}^{-1})$.

		We now show that $\F$ is an $\ell$-diverse representative of $\solutionset$.
		To this end, let $S \in \solutionset$
		and let $A,B$ be two arbitrary sets such that
		$|A \Delta S| \geq \ell$ and $|B \Delta S| \geq \ell$.
		Since $M^* \in \F$, we may assume by symmetry that, say, 
		$|M^* \Delta A| < \ell$, otherwise we are done.
		Note that $|M^* \Delta S| \leq 2\ell$ 
		and that $|A \Delta S| \leq |A \Delta M^*| + |M^* \Delta S| < 3\ell$.
		We say that some coloring~$c$ is \emph{good} for $A, B, S$ if the conditions of \cref{lem:colored-graph-small-syms} are satisfied,
		i.e.\ if
		\begin{align*}
			A \setminus (B\cup S) &\subseteq c^1, &
			B \setminus (A\cup S) &\subseteq c^2, &
			(A \cap B) \setminus S &\subseteq c^3, \text{ and} &
			S \setminus (A\cap B) &\subseteq c^4.
		\end{align*}
		We distinguish between two cases.
		\begin{description}
				\item[\textbf{Case 1: $|M^* \Delta B| < 3\ell$.}]
					Then $|B \Delta S| \leq |B \Delta M^*| + |M^* \Delta S| \leq 5\ell$.
					According to \cref{thm:derandomized-color-coding}
					there is an $i \in \JJ$ such that coloring $c_i$ is good for $A,B,S$,
					since $|B \Delta S| + |A \Delta S| < 8\ell$.
					By \cref{lem:colored-graph-small-syms} and construction of $\F_i$,
					there is an $\wS \in \F_i \subseteq \F$
					such that $|\wS \Delta A| \geq |S \Delta A| \geq \ell$ and
					$|\wS \Delta B| \geq |S \Delta B| \geq \ell$.
					
				\item[\textbf{Case 2: $|M^* \Delta B| \geq 3\ell$.}]
					Set $B':=A$.
					According to \cref{thm:derandomized-color-coding}
					there is an $i \in \JJ$ such that coloring $c_i$ is good for $A,B',S$,
					since $|B' \Delta S| + |A \Delta S| < 6\ell$.
					Thus, by \cref{lem:colored-graph-small-syms} and by the construction of $\F_i$
					there is an $\wS \in \F_i \subseteq \F$
					such that $|\wS \Delta A| \geq |S \Delta A| \geq \ell$.
					Finally, we observe that
					$\abs{\wS \Delta B} \geq \abs{M^* \Delta B} - \abs{M^* \Delta \wS} \geq 3\ell - 2\ell \geq \ell$
					by the triangle inequality.
		\end{description}
		This completes the proof.
\end{proof}
}%
Next, we show how to generate an $\ell$-diverse representative of the family of solutions
if there are two solutions such that no other solution differs from both by more than~$2\ell$.
\newcommand{\famsizethree}{2^{20\ell+o(\ell)} \log n}
\begin{lemma}
		\label{lem:two-matchings}
		Let \cref{assumption} be true.
		Let $I$ be an $\Pi$-instance of size $n$, 
		and $M_1, M_2 \in \solutionset$ such that $\abs{M_1 \Delta M_2} \geq 2\ell$ and
		each $M \in \solutionset$ has $\min\{\abs{M \Delta M_1}, \abs{M \Delta M_2}\} \leq 2\ell$.
		Then one can compute, in $\famsizethree{}\cdot f(\param) n^{\bigO(1)} g(\nicefrac{p}{n^4 \famsizethree{}})$ time and with error probability~$p$, 
		an $\ell$-diverse representative of $\solutionset$ of size $\famsizethree{} \cdot n^4$.
\end{lemma}
\begin{proof}
		For simplicity, let $\JJ := [\famsizethree]$.
		Apply \cref{thm:derandomized-color-coding} with $b = 10\ell$ to compute
		in $\famsizethree \cdot n$~time
		a family of colorings $\{c_j \colon \BB(I) \to \oneto{4} \mid j \in \JJ\}$.
		By \cref{thm:derandomized-color-coding} this family has size $\abs{\JJ}$.
		
		For each $j\in \JJ$, apply \cref{lem:colored-graph-small-syms} to $I$ and $c_j$
		to compute a family~$\F_j \subseteq \solutionset$ of size at most $n^4$
		with error probability~$p/3 \cdot n^{-4}\abs{\JJ}^{-1}$.
		Observe that the probability of an error occurring at any of the~$n^4\abs{\JJ}$ steps is upper-bounded by $p/3$
		and the computation of all~$\F_j$ takes $\abs{\JJ} f(\param) n^{\bigO(1)} g(\nicefrac{p}{3n^4 \cdot \abs{\JJ}})$~time.
		
		Next, define another family of colorings $\{c'_j \colon \BB(I) \to \oneto{2} \mid j \in \JJ\}$
		by setting $c'_j(x) := \ceil{\nicefrac{c_j(x)}{2}}$.
		Then, for each $j \in \JJ$, apply \cref{lem:colored-graph-stay-close-to-M}, to $I$, $c'_j$ and $M_1$
		to compute a family~$\F'_j \subseteq \solutionset$,
		with the same error probability and time bound as before.
		Repeat with $M_2$ instead of~$M_1$ to obtain $\F''_j$.
				
		Set $\F := \{M_1,M_2\} \cup \bigcup_{j \in \JJ} (\F_j \cup F'_j \cup F''_j)$.
		Then $\F$ has size at most~$3\abs{\JJ}n^4+2 \subseteq \famsizethree \cdot n^4$.
		Computing~$\F$ takes $\famsizethree \cdot f(\param) n^{\bigO(1)} g(\nicefrac{p}{3n^4 \cdot \abs{\JJ}})$~time.
		The probability of an error occurring at any step while computing~$\F$ is upper-bounded by $p$.
		
		We now show that $\F$ is an $\ell$-diverse representative of $\solutionset$.
		To this end, let $S \in \solutionset$ and $A, B$ be two arbitrary sets such that
		$\abs{A \Delta S} \geq \ell$ and $\abs{B \Delta S} \geq \ell$.
		We may assume for each $i \in \oneto{2}$ that $\abs{M_i \Delta A} < \ell$ or $\abs{M_i \Delta B} < \ell$,
		otherwise we are done.
		By symmetry, we may assume $\abs{M_1 \Delta A} < \ell$.
		Then $\abs{M_2 \Delta A} \geq \abs{M_2 \Delta M_1} - \abs{M_1 \Delta A} \geq \ell$ by the triangle inequality
		and thus we must have $\abs{M_2 \Delta B} < \ell$.
		By assumption, $\min\{\abs{S \Delta M_1}, \abs{S \Delta M_2}\} \leq 2\ell$,
		so let \wilog{} $\abs{S \Delta M_1} \leq 2\ell$.
		Note that $|A \Delta S| \leq |A \Delta M_1| + |M_1 \Delta S| < 3\ell$.
		We distinguish the following two cases.
		\begin{description}
				\item[\textbf{Case 1: $|M_1 \Delta M_2| \leq 4\ell$.}]
						Then, $|B \Delta S| \leq |B \Delta M_2| + |M_2 \Delta M_1| + |M_1 \Delta S| < 7\ell$.
		We say that some coloring~$c$ is \emph{good} for $A, B, S$ if the conditions of \cref{lem:colored-graph-small-syms} are satisfied,
		i.e.\ if
				$
			A \setminus (B\cup S) \subseteq c^1,\,
			B \setminus (A\cup S) \subseteq c^2,\,
			(A \cap B) \setminus S \subseteq c^3, \text{ and }
			S \setminus (A\cap B) \subseteq c^4.
			$
						According to \cref{thm:derandomized-color-coding}
						there is an $i \in \JJ$ such that coloring $c_i$ is good for $A,B,S$,
						since $|B \Delta S| + |A \Delta S| \leq 10\ell$.
						By \cref{lem:colored-graph-small-syms}, there is $\wS \in \F_i \subseteq \F$ such that 
						such that $|\wS \Delta A| \geq |S \Delta A| \geq \ell$ and
						$|\wS \Delta B| \geq |S \Delta B| \geq \ell$.
				\item[\textbf{Case 2: $|M_1 \Delta M_2| > 4\ell$.}]
						Since $|S \Delta A| \leq 3\ell \leq 10\ell$,
						there is $j \in \JJ$ such that
						$
							S \setminus A \subseteq {c'_j}^1 \text{ and }
							A \setminus S \subseteq {c'_j}^2.
						$
						By \cref{lem:colored-graph-stay-close-to-M} there is $\wS \in \F_j'$
						such that 
						$\abs{\wS \Delta M_1} = \abs{S \Delta M_1} \leq 2\ell$ and $\abs{\wS \Delta A} \geq \abs{S \Delta A} \geq \ell$.
						Finally, observe that by the triangle inequality
						$\abs{\wS \Delta B} \geq \abs{M_1 \Delta M_2} - \abs{M_1 \Delta \wS} - \abs{B \Delta M_2} > \ell$.
		\end{description}
		This completes the proof.
\end{proof}

\noindent With \cref{lem:one-matchings,lem:two-matchings,lem:three-matchings} at hand
we can formalize the case distinction outlined in the beginning of the section.
This gives us a way to efficiently compute an $\ell$-diverse representative in general.
\begin{lemma}
		\label{lem:compute-diverse-F}
		Let \cref{assumption} be true.
		Let $I$~be an instance of $\Pi$ of size $n$.
		One can compute an~$\ell$-diverse representative of $\solutionset$ of size $\famsizethree \cdot n^4$ 
		in $\famsizethree \cdot f(\param) n^{O(1)} g(\nicefrac{p}{n^4 \cdot \famsizethree})$ time
		with error probability at most~$p$.
\end{lemma}
\begin{proof}
		Our procedure to compute an $\ell$-diverse representative of $\solutionset$ works in four steps.

		\begin{description}
			\item[Step 1.]
		We use \Pialgo{} with a monochrome coloring and error probability~$\nicefrac{p}{4n}$ to search for 
		some $M_1 \in \solutionset$ in $f(\param) n^{\bigO(1)} g(\nicefrac{p}{4n})$ by guessing the size of $|M_1| \leq n$. 
		Observe that the probability of an error occurring in any of the searches is upper-bounded by $\nicefrac{p}{4}$
		If we do not succeed, then output the empty set and we are done.
		Otherwise,  we proceed with the next step.

			\item[Step 2.]
		For each pair $m_1, m_2$ with $m_1 + m_2 \leq n$ and $m_2 + \abs{M_1} - m_1 > 2\ell$,
		try to compute $M_2 \in \solutionset$ with $\abs{M_2 \cap M_1} = m_1$ and $\abs{M_2 \cap (\BB(I) \setminus M_1)} = m_2$
		in $f(\param)n^{\bigO(1)} g(\nicefrac{p}{4n^2})$~time and with error probability~$\nicefrac{p}{4 n^2}$
		using \Pialgo{} with a $2$-coloring where elements in $M_1$ are assigned one color and elements in $\BB(I) \setminus M_1$ are assigned the second color.
		If no such~$M_2$ is found for any pair~$m_1, m_2$, then for every $M \in \solutionset$ the symmetric difference $\abs{M \Delta M_1} \leq 2\ell$.
		In that case we may apply \cref{lem:one-matchings} with error probability $p/2$ and are done.
		Observe that the probability of an error occurring at any step until here is upper-bounded by $p$
		and the overall running time is $\famsizetwo \cdot f(\param) n^{O(1)} g(\nicefrac{p}{n^4 \cdot \famsizetwo})$.
		If we found such an $M_2$, then we proceed with the next step.

			\item[Step 3.]
		We have $M_1,M_2 \in \solutionset$ with $\abs{M_1 \Delta M_2} \geq 2\ell$.
		Define the coloring~$c \colon \BB(I) \to \oneto{4}$ by
		\begin{align*}
				c(v) := 
			\begin{cases}
					i & \text{ if } v \in M_i \setminus M_j \text{ for } \{i,j\}=\{1,2\}, \\
					3 & \text{ if } v \in M_1 \cap M_2, \text{ and}\\
					4 & \text{ otherwise.}
			\end{cases}
		\end{align*}
		For all $m'_1, m'_2, m'_3, m'_4$ 
		with $m'_1 + m'_2 + m'_3 + m'_4 \leq n$
		and $m'_2+m'_4 + \abs{M_1} - m'_1 - m'_3 \geq 2\ell$ 
		and $m'_1+m'_4 + \abs{M_2} - m'_2 - m'_3 \geq 2\ell$,
		search for a solution $M_3 \in \solutionset$
		with $|M_3 \cap c^i| = m'_i$, for all $i \in \oneto{4}$, 
		using \Pialgo{} with $c$ and error probability $\nicefrac{p}{4n^4}$.
		For all these combined, we thus have error probability~$p/4$ and 
		need $f(\param)n^{\bigO(1)} g(\nicefrac{p}{4n^4})$~time.
		If no such~$M_3$ is found for any choice of $m'_1, m'_2, m'_3, m'_4$, 
		then any $M \in \solutionset$ must have $\min\{\abs{M \Delta M_1}, \abs{M \Delta M_2}\} < 2\ell$.
		In that case we may apply \cref{lem:two-matchings} with error probability $\nicefrac{p}{4}$ and are done.
		Observe that the probability of an error occurring at any step until here is upper-bounded by $p$
		and the overall running time is $\famsizethree \cdot f(\param) n^{O(1)} g(\nicefrac{p}{n^4 \cdot \famsizethree})$.
		In case that we found such an $M_3$, we proceed with the next step.

		\item[Step 4.]
		We have $M_1,M_2,M_3 \in \solutionset$ such that $|M_i \Delta M_j| \geq 2\ell$ for all distinct $i,j \in \oneto{3}$.
		Hence, by \cref{lem:three-matchings}, we can output $\{M_1,M_2,M_3\}$.
		This completes the proof.\qedhere
		\end{description}
\end{proof}

\ifarxiv{}Finally, \cref{lem:compute-diverse-F} allows us to formulate a dynamic program for \dmPi{} and 
prove \cref{thm:diverse-multistage-pi}.\else{}
Finally, \cref{lem:compute-diverse-F} allows us to prove \cref{thm:diverse-multistage-pi}, see \cref{proof:thm:diverse-multistage-pi}.\fi{}
		\appendixproof{thm:diverse-multistage-pi}{
\begin{proof}[Proof of \cref{thm:diverse-multistage-pi}]
		Let $J := ((I_i)_{i=1}^\lifetime,\ell)$ be an instance of \dmPi{}, where $n := \max_{i \in [\lifetime]} |I_i|$.
		For each $i \in [\lifetime]$ we apply \cref{lem:compute-diverse-F} to obtain
		an $\ell$-diverse representative~$\F_i$ of $\solutionset[I_i]{}$
		that has size at most $\famsizethree \cdot n^4$
		in $\famsizethree \cdot f(\param)n^{\bigO(1)} \cdot g(\nicefrac{p}{\lifetime n^4 \cdot \famsizethree})$~time
		with error probability $\nicefrac{p}{\lifetime}$. 
		Observe that the probability of an error occurring at any step is upper-bounded by $p$.
		Now we use the following dynamic program to check whether $J$ is a \yes-instance.
		\begin{align*}
				\forall i \in \{2, 3, \dots, \lifetime\}, S \in \F_{i} \colon
				D[i,S] \!:=\!
				\begin{cases}
					\top & \text{if }\exists \wS \in \F_{i-1} \colon D[i-1,\wS] = \top \text{ and } 
								|S \Delta \wS| \geq \ell,\\
						\bot & \text{otherwise,}
				\end{cases}
		\end{align*}
		where $D[1,\wS] = \top$ if and only if $\wS \in \F_1$.
		We report that $J$ is a \yes-instance if and only there is an $S \in \F_\lifetime$ such that $D[\lifetime,S] = \top$.
		Note that this takes $\left(\famsizethree \cdot n^4 \right)^2 \lifetime \subseteq 2^{\bigO(\ell)} n^{\bigO(1)} \lifetime$ time. 
		Hence our overall running time is~$2^{\bigO(\ell)} \cdot f(\param_{\max})n^{\bigO(1)} \cdot \lifetime \cdot g(\nicefrac{p}{\lifetime n^4 \cdot \famsizethree})$, where~$r_{\max}$ is the maximum of parameter $r$ over all instances of $\Pi$ in $J$.

		\textbf{($\Leftarrow$):} We show by induction over~$i \in [\lifetime]$ that 
		if $D[i,S] = \top$, then there is a sequence $(S_j)_{j \in [i]}$ 
		such that $S_i = S$, $S_j \in \solutionset[I_j]{}$ for all $j \in [i]$
		and $|S_{j-1} \Delta S_j| \geq \ell$ for all $j \in \{2, 3, \dots, i\}$.

		By definition of $D$ this is clearly the case for $i=1$.
		Now let $1 < i \le \lifetime$ and $D[i,S] = \top$.
		Since $D[i,S] = \top$, $S \in \F_i$ and thus $S \in \solutionset[I_i]{}$.
		By definition of $D$ there is an $\wS \in F_{i-1}$ with $D[i-1,\wS]=\top$ and $|S \Delta \wS| \geq \ell$.
		By induction hypothesis, there is a sequence $(S_j)_{j \in [i-1]}$ 
		such that $S_{i-1} = \wS$, $S_j \in \solutionset[I_j]{}$ for all $j \in [i-1]$
		and $|S_{j-1} \Delta S_j| \geq \ell$ for all $j \in \{2, 3, \dots, i-1\}$.
		Hence, the sequence $(S_1,\dots,S_{i-1}=\wS,S)$ completes the induction.
		Thus, if we report that $J$ is a \yes-instance, then this is true.

		\textbf{($\Rightarrow$):} 
		Now let $(S_j)_{j\in[\lifetime]}$ be a solution for $J$.
		To simplify the proof let $S_{\lifetime+1}$ be a set of $\ell$~elements that are disjoint from~$S_{\lifetime}$.
		We show by induction that for all $i \in [\lifetime]$
		there is a $Z \in \F_i$ such that $D[i,Z] = \top$ and $|Z \Delta S_{i+1}| \geq \ell$.

		Let $i=1$.
		Then there is a $Z \in \F_1$ such that $|S_2 \Delta Z| \geq \ell$ 
		since $\F_1$ is an $\ell$-diverse representative of $\solutionset[I_1]$.
		Hence, $D[1,Z] = \top$.

		Now let $1<i\le \lifetime$.
		By induction hypothesis, there is a $Z_{i-1} \in \F_{i-1}$ such that $D[i-1,Z_{i-1}] = \top$ and
		$|S_i \Delta Z_{i-1}| \geq \ell$. 
		Since $S_i \in \solutionset[I_i]$ and we have $|S_i \Delta Z_{i-1}|,|S_i \Delta S_{i+1}|\geq \ell$ 
		and $\F_i$ is an $\ell$-diverse representative of $\solutionset[I_i]$,
		there is a $Z \in \F_i$
		such that $|Z \Delta Z_{i-1}|,|Z \Delta S_{i+1}|\geq \ell$.
		By definition of $D$, we also have $D[i,Z]=\top$. 
		This completes the induction step.
		Thus, there is a $Z\in \F_{\lifetime}$ such that $D[\lifetime, Z] =\top$
		and if $J$ is a \yes-instance, then we report that.
\end{proof}
}%

\section{Application: Committee Election}
\label{sec:committee}
\appendixsection{sec:committee}

\citet{bredereck2020multistage} studied the following problem
under the name \textsc{Revolutionary Multistage Plurality Voting}.
\problemdef{\dmVoting}
{A set $A$ of agents, 
 a set $C$ of candidates,
 a sequence $(u_i)_{i=1}^\lifetime$ of voting profiles $u_i \colon A \to C \cup \{ \emptyset \}$, and
 integers $k,x,\ell \in \mathbb N$.}
 {Is there a sequence $(C_1,C_2,\dots,C_\lifetime)$ such that
	for all $i \in [\lifetime]$ it holds that $|C_i| \leq k$
	and $|u^{-1}(C_i)| \geq x$, 
	and for all $i \in [\lifetime-1]$ it holds true that $\abs{C_i \Delta C_{i+1}} \geq \ell$?}
In this section, we affirmatively answer the question of \citet{bredereck2020multistage} whether
\dmVoting{} parameterized by $\ell$ or $k$ is in \FPT{}.%
\footnote{Note that $\ell \leq 2k$ for all non-trivial instances, so it suffices to prove this for~$\ell$.}

\begin{theorem}
		\label{thm:dm-voting}	
		An instance $J$ of \dmVoting{} can be solved in $2^{\bigO(\ell)}\cdot |J|^{\bigO(1)}$ time. 
\end{theorem}
To prove \cref{thm:dm-voting}, we use \cref{thm:diverse-multistage-pi-det}.
In the notation of our framework, we deal with the following problem~$\Pi$: 
given an instance $I=(A,C,u,k,x)$ consisting of 
a set $A$ of agents,
a set $C =: \BB(I)$ of candidates,
a voting profile $u \colon A \to C$, and
two integers $k,x$,
decide whether $\solutionset := \{ S \subseteq C \mid k \geq |S| \text{ and } |u^{-1}(S)| \geq x \}$ is non-empty.
Hence, to apply \cref{thm:diverse-multistage-pi-det}, 
we consider the following problem.
\optproblemdef{\coloredVoting}
{A set $A$ of agents, 
 a set $C$ of candidates, 
 a voting profile $u \colon A \to C \cup \{ \emptyset \}$,
 a coloring $c\colon C \to [4]$, and
 integers $n_i,x,k \in \mathbb N,i\in[4]$.}
 {A set $C' \subseteq C$ of at most $k$ candidates
so that $|u^{-1}(C')| \geq x$ and $|c^{-1}(i) \cap C'|=n_i$ for all $i \in [4]$
or ``\no'' if no such set exists.}
This problem is polynomial-time solvable and 
hence the following observation together with \cref{thm:diverse-multistage-pi-det} proves \cref{thm:dm-voting}.
In \cref{sec:matroids} we will generalize this application to independent sets in matroids.
\ifarxiv{}
\begin{observation}
\else{}
\begin{observation}
		[\appref{lem:coloredVoting-polytime}]
\fi{}
		\label{lem:coloredVoting-polytime}
		\coloredVoting{} is polynomial-time solvable.
\end{observation}
\appendixproof{lem:coloredVoting-polytime}{
\begin{proof}
		Given an instance $I=(A,C,u,c,n_1,n_2,n_3,n_4,x,k)$ of \coloredVoting{}.
		We may assume that $\sum_{i=1}^4 n_i \leq k$, otherwise we can terminate without an output.
For each candidate $v \in C$, we compute its \emph{score} $s(v) := |u^{-1}(v)|$.
		Let $C_i := c^{-1}(i)$, for all $i \in [4]$.
		For each $i \in [4]$, sort the candidates in $C_i$ by their scores.
		Compute the set $C_i' \subseteq C$ containing the $n_i$ candidates of $C_i$ with the highest score.
		If $\sum_{i=1}^4\sum_{v \in C_i'} s(v) \geq x$, then output $\bigcup_{i=1}^4 C_i$.
		Otherwise, we terminate without an output.
		It is easy to verify that this procedure is correct.
\end{proof}
}%

\section{Application: Perfect Matching}
\label{sec:matching}
\appendixsection{sec:matching}

In this section, we apply our framework from \cref{sec:framework} to find a sequence of diverse perfect matchings.

\problemdef{\dmMatching}
{A sequence $(G_i)_{i=1}^\tau$ of graphs and an integer~$\ell \in \NN_0$.}
{Is there a sequence~$(M_i)_{i=1}^\tau$ of perfect matchings~$M_i \subseteq E(G_i)$ such that $|M_{i}\Delta M_{i+1}| \ge \ell$ for all~$i \in \oneto{\tau -1}$?}

There are two closely related variants of this problem which were studied extensively.
The first variant is the non-diverse variant, where one seeks to bound the symmetric differences (in some way) from above~\cite{bampis2018multistage,bampis2019lp,chimani2020approximating,gupta2014changing,steinhau2020parameterized}.
\citet{steinhau2020parameterized} proved that 
if the size of the symmetric difference of two consecutive perfect matchings shall be at most $\ell$,
then this problem variant is \NP-hard even if $\ell$ is constant, 
and~$\Wone$-hard when parameterized by~$\ell+\tau$.
The second variant is the non-multistage variant, where one is given a single graph and is asked to compute a set of pairwise diverse perfect matchings~\cite{fomin2020diverse,FominGPP021}.
\citet{fomin2020diverse} proved that this variant is~$\NP$-hard even if one asks only for two diverse matchings.
This directly implies~$\NP$-hardness for \dmMatching{} even when~$\tau = 2$.

Our goal is to show fixed-parameter tractability of \dmMatching{} when parameterized by~$\ell$.
This stands in contrast to the $\NP$-hardness for the non-diverse problem variant with constant~$\ell$.

\ifarxiv{}
\begin{theorem}
\else{}
\begin{theorem}
		[\appref{thm:dmmatching}]
\fi{}
	\label{thm:dmmatching}
	An instance $J$ of \dmMatching{} can be solved in $2^{\bigO(\ell)} \cdot |J|^{\bigO(1)}$ time
	with a constant error probability.
\end{theorem}
We will prove \cref{thm:dmmatching} by means of \cref{thm:diverse-multistage-pi} at the end of this section.
To this end we need to consider the following problem.

\optproblemdef{\ExactMatchingProb}{
	A graph $G=(V,E)$, a coloring~$c \colon E \to \oneto{s}$, and $k_i\in\NN,i\in [s]$.
	}{
	(if exists) A perfect matching $M$ in $G$ such that $|c^i \cap M| = k_i$, for all $i \in [s]$?
}

For~$s=2$, this problem is known as \textsc{Exact Matching}, and \citet{MulmuleyVV87} showed that this special case is solvable by a randomized polynomial-time algorithm.
We generalize this result by showing that \ExactMatchingProb{} can be solved in polynomial time for any constant~$s$ by a randomized algorithm with constant error probability.
While we only need this for~$s=4$ in order to prove \cref{thm:dmmatching}, we believe that the general case may be of independent interest.
We remark that it is open whether \textsc{Exact Matching} can be solved in (deterministic) polynomial time.
\ifarxiv{}
\begin{lemma}
\else{}
\begin{lemma}
		[\appref{thm:findColoredMatching}]
\fi{}
		\label{thm:findColoredMatching}
		For every~$0 < p < 1$
		there is an~$(n^{\bigO(s)} \cdot \log 1/p)$-time algorithm which, given an instance of \ExactMatchingProb{},
	finds a solution with probability at least~$1-p$ if one exists,
	and concludes that there is no solution otherwise.
\end{lemma}
\ifarxiv{}The proof of \cref{thm:findColoredMatching} deferred for a moment.\fi{}
To determine whether a given \ExactMatchingProb{} has a solution we use the following algorithm.
\begin{algorithm}
		\label{alg:matching}
		Let $0 < p < 1$
		and let $I = (G,c,k_1,\dots,k_s)$ be an instance of \ExactMatchingProb{}
		where $G=(V,E)$ has $n$~vertices.
	\begin{description}
			\item[Step 1.] Set $\gamma := \lceil n/(2p)\rceil$ and 
					draw~$w_{ij} \in [\gamma]$ for all~$\{i,j\} \in E$ uniformly at random.
			\item[Step 2.] Construct an $n\times n$ matrix $A'$ with entries~$a_{ij} \in \ZZ[y_1, \dots, y_s]$, $1 \le i \le j \le n$, where
\begin{align*}
	a_{ij} := \begin{cases}
		0 & \text{if } \{i,j\} \notin E,\\
		w_{ij}y_q & \text{if } \{i, j\} \in c^q \cap E, q \in [s].\\
	\end{cases}
\end{align*}
Afterwards we compute the  skew-symmetric matrix~$A := A' - (A')^T$.
			\item[Step 3.] Compute the polynomial $P := \sqrt{\det(A)} \in \ZZ[y_1,\dots,z_s]$.

			\item[Step 4.] If $P$ contains a monomial $b^*y_1^{k_1} y_2^{k_2} \cdots y_s^{k_s}$ 
					such that $b^*\neq 0$ then, output \yes. Otherwise, output \no.  
					\hfill$\diamond$
	\end{description}
\end{algorithm}
Before studying the running time of~\cref{alg:matching}, we first focus on its correctness.
\begin{lemma}
		\label{lem:exactmatching-probability}
		Let $I$ and $p$ be the input of \cref{alg:matching}.
	If \cref{alg:matching} returns \yes, then there is a solution for $I$.
	Conversely, if $I$ is a \yes-instance, then \cref{alg:matching} returns \yes{} with probability at least~$1-p$.
\end{lemma}
\toappendix{
To show \cref{lem:exactmatching-probability} we need the following well-known lemma.
\begin{lemma}[\citet{DL78,Schwartz80,Zippel79}]
	\label{lem:schwartz-zippel}
	Let~$P \in \mathbb F[x_1, \dots, x_n]$ be a polynomial of total degree~$d \ge 0$ over a field~$\mathbb F$.
	Let~$S$ be a finite subset of~$F$ and let~$r_1, \dots, r_n$ be selected uniformly and independently at random from~$S$.
	Then the probability that~$P(r_1, \dots, r_n) = 0$ is at most~$d/|S|$.
\end{lemma}
}
\ifarxiv{}
\begin{proof}[Proof of \cref{lem:exactmatching-probability}]
\else{}
\begin{proof}
\fi
	Let~$\mathcal P$ be the set of all partitions of~$V$ into unordered pairs.
	For~$\sigma \in \mathcal P$ with~$\sigma = \{\{i_1, j_1\}, \{i_2, j_2\}, \dots, \{i_{n/2}, j_{n/2}\}\}$ with~$i_k < j_k$ for~$k \in [n/2]$ and~$i_1 < i_2 < \dots < i_{n/2}$, let
	\[
		\pi_\sigma := \begin{bmatrix}
			1&2&3&4&\cdots&n-1&n\\
			i_1&j_1&i_2&j_2&\cdots&i_{n/2}&j_{n/2}
		\end{bmatrix}
	\]
	be the corresponding permutation.
	Let~$\val(\sigma) := \sgn(\pi_\sigma)\prod_{\{i,j\} \in \sigma} a_{ij}$, where~$\sgn(\pi_\sigma) \in \{+1, -1\}$ is the signum of~$\pi_\sigma$.
	The Pfaffian of~$A$ (computed by \cref{alg:matching}) is defined as~$\pf(A) := \sum_{\sigma \in \mathcal P} \val(\sigma)$ \cite{lovasz2009matching}.
	Note that~$A$ is skew-symmetric, hence, $\pf(A) = \sqrt{\det(A)} = P$ \cite{muir1882,lovasz2009matching}.
	As~$\val(\sigma) = 0$ whenever~$\sigma$ contains a non-edge, we have~$P = \sum_{M \in \mathcal{PM}} \val(M)$, where~$\mathcal{PM}$ is the set of perfect matchings in $G$.
	Let~$M$ be a perfect matching and let~$z_q = |c^q \cap M|$, $q \in \oneto{s}$.
	Then
	$
		\val(M) = \sgn(\pi_M) \prod_{q \in \oneto{s}} \: \prod_{\{i,j\} \in M \cap c^{q}} w_{ij}y_q = b \cdot y_1^{z_1} y_2^{z_2} \cdots y_s^{z_s},
	$
	where~$b \in \ZZ$.
	Let $\mathcal{PM^*} \subseteq \mathcal{PM}$ be the family of perfect matchings $M^*$ which have exactly $k_i$~edges of color~$i$, for all~$i \in [s]$.
	Then the coefficient~$b^*$ of the monomial~$b^*y_1^{k_1} y_2^{k_2} \cdots y_s^{k_s}$ of~$P$
	is
	$
		b^* = \sum_{M^* \in \mathcal{PM^*}} \sgn(\pi_{M^*})\prod_{\{i,j\} \in M^*} w_{ij}\,.
	$
	Hence, if \cref{alg:matching} returns \yes{} (i.e., $b^* \neq 0$), then $\mathcal{PM}^* \neq \emptyset$.
	
	Now conversely assume $I$~to be a \yes-instance, i.e., $\mathcal{PM^*} \neq \emptyset$.
	We analyze the probability of the event $b^*=0$ occurring.
	Note that $b^*$ can be seen as a polynomial of degree at most~$n/2$ over the indeterminates $\{w_{ij} \mid \{i,j\} \in E\}$.
	As we have drawn the~$w_{ij}$ independently and uniformly at random from~$[\gamma]$ with~$\gamma \ge n/(2p)$,
	by the DeMillo-Lipton-Schwartz-Zippel lemma (\cref{lem:schwartz-zippel}) the probability that~$b^* = 0$ is at most~$n/(2\gamma) \le p$.
\end{proof}
Now we show that \cref{alg:matching} can be executed efficiently.
\ifarxiv{}
\begin{lemma}
\else{}
\begin{lemma}
		[\appref{lem:exactmatching-time}]
\fi{}
		\label{lem:exactmatching-time}
		\cref{alg:matching} runs in $n^{O(s)} \log(1/p)$ time.
\end{lemma}
\appendixproof{lem:exactmatching-time}{
As for the running time of \cref{alg:matching}, note that computing the determinant as well as its square root are the most expensive operations.
For completeness, 
We first show that we can compute the square root of a polynomial efficiently.
\begin{lemma}
	\label{lem:rootOfPolynomials}
	Let~$P \in \ZZ[x_1, \dots, x_s]$ be a polynomial of degree~$2n>s$ such that there exists a polynomial~$Q \in \ZZ[x_1, \dots, x_s]$ with~$P(x) = (Q(x))^2$.
	Computing~$Q$ from~$P$ takes~$n^{\bigO(s)}$ algebraic operations.%
	\footnote{That is, additions, subtractions, multiplications, divisions.}
\end{lemma}
\begin{proof}
	For $\alpha \in \NN^s$, we write $x^\alpha$ to denote $x_1^{\alpha_1}x_2^{\alpha_2} \cdots x_s^{\alpha_s}$.
	Let~$P(x) =: \sum_{\alpha \in \NN^s} c_\alpha x^\alpha$
	and $Q(x) =: \sum_{\alpha \in \NN^s} d_\alpha x^\alpha$
	with coefficients $c_\alpha, d_\alpha \in \ZZ$.
	Then
	\begin{align*}
		P(x) &= \sum_{\alpha, \beta \in \NN^s} d_{\alpha}d_{\beta}x^{\alpha+\beta}
		\intertext{and thus for all $\kappa \in \NN^s$}
		c_\kappa &= \sum_{\substack{\alpha, \beta \in \NN^s\\\alpha + \beta = \kappa}} d_\alpha d_\beta \,.
	\end{align*}
	We will compute the~$d_\kappa$ by induction on~$\norm{\kappa}_1$.
	Clearly, $d_{(0,0,\dots,0)} = \sqrt{c_{(0,0,\dots,0)}}$, this is our base case.
	Now let~$\kappa \in \NN^s$ and suppose we have computed all~$d_\alpha$ with~$\norm{\alpha}_1 < \norm{\kappa}_1$.
	We have
	\[
		c_{\kappa} - \sum_{\substack{\alpha + \beta = \kappa\\ \norm{\alpha}_1 \ne 0 \\ \norm{\beta}_1 \ne 0}} d_{\alpha}d_{\beta}
			= \sum_{\substack{\alpha + \beta = \kappa\\ \norm{\alpha}_1 \cdot \norm{\beta}_1 = 0}} d_{\alpha}d_{\beta}
			= 2 \cdot d_{(0,0,\dots,0)} \cdot d_\kappa.
	\]
	This is equivalent to
	\begin{align}\label{eq:coefficient}
		d_{\kappa} = \frac{1}{2 \cdot d_{(0,0,\dots,0)}} \cdot \Big(c_{\kappa} - \sum_{\substack{\alpha + \beta = \kappa\\ \norm{\alpha}_1 \ne 0 \\ \norm{\beta}_1 \ne 0}} d_{\alpha}d_{\beta} \Big).
	\end{align}
	We already computed~$d_{(0,0,\dots,0)}$ as well as all~$d_\alpha$ and~$d_\beta$ occuring in~\eqref{eq:coefficient}.
	Therefore we can use~\eqref{eq:coefficient} to compute~$d_\kappa$ and thus $Q(x)$.

	Note that the sum in~\eqref{eq:coefficient} contains one summand for each $d_\alpha$ with $0 < \norm{\alpha}_1 < \norm{\kappa}_1$.
	Since $1/(2d_{(0,0,\dots,0)})$ only needs to be computed once, computing $d_\kappa$ requires
	\[
		1 + \sum_{j=1}^{\norm{\kappa}_1-1} \binom{s+j}{j} = \sum_{j=0}^{\norm{\kappa}_1-1} \binom{s+j}{j} = \binom{s+\norm{\kappa}_1}{s+1}
	\]
	algebraic operations.
	As we need to do this for every~$\kappa \in \NN^s$ with~$\norm{\kappa}_1 \le n$, we require
	\[
		\bigO(1) + \sum_{i=1}^{n} \binom{s+i}{s} \cdot \binom{s+i}{s+1}
		\leq \bigO\left(n \cdot (s+n)^{s} \cdot (s+n)^{s+1}\right)
		\leq \bigO\left(n^{2s+2}\right)
		\leq n^{\bigO(s)}
	\]
	algebraic operations overall.
\end{proof}
\begin{proof}[Proof of \cref{lem:exactmatching-time}]
		Note that, \wilog{}, $s \le n/2$.
		Moreover, we need at most $O(\log{\gamma})$ cells to store $w_{ij}$ for all $\{i,j\} \in E$.
		Hence, computing~$A$ takes~$\bigO(n^2\log{\gamma})$ time.
	As~$\det(A)$ is a polynomial of degree at most~$n$ with~$s$ variables, it consists of at most~$\binom{n+s}{s} \in \bigO((2n)^s)$ coefficients.
	Hence, computing~$\det(A)$ takes at most~$n^{\bigO(s)}$ algebraic operations, e.g., using Gauss elimination. 
	As we need at most $O(\log{\gamma})$ cells for the initial values $w_{ij}$ for all $\{i,j\} \in E$,
	and we need only~$n^{\bigO(s)}$ algebraic operations to compute~$\sqrt{\det(A)}$ (see \cref{lem:rootOfPolynomials}),
	we have an overall running time of $n^{\bigO(s)} \log(\gamma) = n^{\bigO(s)} \log{(1/p)}$.
\end{proof}
}%
We are now ready to put all parts together and prove \cref{thm:findColoredMatching}.
In a nutshell, we use \cref{alg:matching} to check whether there is a solution.
If this is the case, then we try to delete as many edges as possible from the instance until the whole edge set is a solution\ifarxiv{}\else{}, see \cref{proof:thm:findColoredMatching} for details\fi{}.
\appendixproof{thm:findColoredMatching}{
\begin{proof}[Proof of \cref{thm:findColoredMatching}]
	Let $I = (G=(V,E),c,k_1,\dots,k_s)$ be an instance of \ExactMatchingProb{}.
	Let $m=|E|$ and $n=|V|$.
	We check whether $I$ has a solution by applying \cref{alg:matching} with error probability $p/(m+1)$.
	If the answer is \no{}, then we output that there is no solution.
	Otherwise we initialize $M := \emptyset$.

	We iterate over all edges 
	$e \in E$ and 
	apply \cref{alg:matching} with error probability $p/(m+1)$
	to the instance $((V,M \cup (E \setminus \{e\})), c, k_1, \dots, k_s)$.
	Afterwards, we delete~$e$ from~$E$.
	If the result is \no{}, then we add $e$ to~$M$.
	In any case we proceed with the next edge.
	If we reached $M = E$, then we output the solution~$M$.
	
	By \cref{lem:exactmatching-probability}, the probability that at some step an error occurs is at most $p/(m+1)$.
	Since we execute \cref{alg:matching} at most $m+1$ times, the overall error probability is $p$. 

	Overall, we execute \cref{alg:matching} at most $m+1$ times, 
	each of which can be computed in~$n^{\bigO(s)}\log(1/p)$ time due to \cref{lem:exactmatching-time}.
	Hence, the overall running time can be bounded by~$n^{\bigO(s)}\log(1/p)$.
\end{proof}
}%
Putting \cref{thm:findColoredMatching} and \cref{thm:diverse-multistage-pi} together, we can prove the main theorem of this section\ifarxiv{}\else{}, see \cref{proof:thm:dmmatching} for details\fi{}.
\appendixproof{thm:dmmatching}{
\begin{proof}[Proof of \cref{thm:dmmatching}]
	\cref{thm:findColoredMatching} with $s=4$ fulfills \cref{assumption} wherein~$g(p)=\bigO(\log 1/p)$ and~$f(r)=1$.
	We aim for a constant error probability, say $p=1/4$.
	Hence, by \cref{thm:diverse-multistage-pi} and \cref{thm:findColoredMatching} 
	we have an algorithm with error probability~$1/4$ for an instance $J$ of \dmMatching{} with running time
			$2^{\bigO(\ell)} \cdot |J|^{\bigO(1)}$.
\end{proof}
}%

\section{Application: $s$-$t$ Path}
\label{sec:path}
\appendixsection{sec:path}
In this section, we apply our framework to the task of finding a sequence of diverse $s$-$t$ paths.
This has obvious applications e.g.\ in %
convoy routing \cite{FNSZ20}. 

\problemdef{\dmPath}
{A sequence of %
		graphs $(G_i)_{i=1}^\lifetime$, 
	two distinct vertices~$s,t \in \bigcap_{i=1}^\lifetime V(G_i)$,
 and $\ell \in  \NN_0$.}
 {Is there a sequence $(P_1,P_2,\dots,P_\lifetime)$ such that
		 $P_i$ is an $s$-$t$ path in $G_i$ for all $i \in [\lifetime]$,
	and  $\abs{S_i \Delta S_{i+1}} \geq \ell$ for all  $i \in \oneto{\tau -1}$?}
	Our goal is to show that \dmPath{} parameterized by $\ell$ is in \FPT.
\begin{theorem}\label{thm:diverse-paths}
		\dmPath{} parameterized by $\ell$ is in \FPT{}.
\end{theorem}
We will prove \cref{thm:diverse-paths} by means of \cref{thm:diverse-multistage-pi-det} at the end of this section.
To this end, we need to consider the following problem.
\optproblemdef{\coloredPath}
{ A %
		graph $G$,
		distinct vertices~$s,t \in V(G)$, 
		coloring $c \colon V(G) \to [4]$,
		and $n_i \in \NN_0,i\in[4]$.
		}{(if exists) An $s$-$t$ path $P$ such that $\abs{c^{-1}(i) \cap V(P)} = n_i$ for all $i\in[4]$.}

Unfortunately, \coloredPath{} is unlikely to be polynomial-time solvable,
as it is NP-hard even if only a single color is used, by a trivial reduction from \textsc{Hamiltonian Path}.
However, as we will see in the proof of \cref{thm:diverse-paths}, 
by a result of \citet{MNSW-long-cycles} we can actually reduce \coloredPath{} to the case that all graphs have small treewidth.
In this setting, we then employ dynamic programming.

\ifarxiv{}
\begin{lemma}
\else{}
\begin{lemma}
		[\appref{lem:tw-algo}]
\fi{}
		\label{lem:tw-algo}
		\coloredPath{} is solvable $k^{\bigO(k)}\cdot|I|^{\bigO(1)}$ time, 
		where $k$ is the treewidth of the input graph~$G$.
\end{lemma}
While some techniques \cite{cygan2011solving,fomin2014representative,bodlaender2015deterministic} seem applicable to improve the running time of \cref{lem:tw-algo} slightly,
for our needs a straight-forward dynamic program on a nice tree decomposition suffices.
\appendixproof{lem:tw-algo}{
We introduce nice tree decompositions before we prove \cref{lem:tw-algo}.
\newcommand{\treedec}{\TT = (T,\{X_v\}_{v \in V(T)})}
\begin{definition}
		A \emph{tree decomposition} of a graph $G$ is a pair $\TT = (T,\{X_v\}_{v \in V(T)})$,
		where $T$ is a tree whose every node $v$ is assigned a \emph{bag} $X_v \subseteq V(G)$
		such that 
		\begin{enumerate}
				\item  $\bigcup_{v \in V(T)} X_v = V(G)$,
				\item  $\forall e\in E(G),\exists v\in V(T) \colon e \subseteq X_v$,
				\item $\forall u\in V(G)\colon T[\{ v \in V(T) \mid u \in X_v\}]$ is a tree.
		\end{enumerate}
		The \emph{width} of $\TT$ is $\max_{v\in V(T)} |X_v|-1$.
		The \emph{treewidth} of a graph $G$ is the minimum width of a tree decomposition of $G$.
		A tree decomposition $\TT$ is called \emph{nice} if $T$ is rooted at a vertex~$r$ such that
		\begin{itemize}
				\item $X_r=\emptyset= X_v$ for all leaves $v \in V(T)$, and
				\item every non-leaf node $v \in V(T)$ of $T$ is of one of the following three types:
						\begin{description}
								\item{\textbf{Introduce node:}} $v$ has exactly one child $w$ in $T$
										and $X_v = X_w \cup \{u \}$, for some vertex $u \not\in X_w$.
								\item{\textbf{Forget node:}} $v$ has exactly one child $w$ in $T$
										and $X_v = X_w \setminus \{u \}$, for some vertex $u \in X_w$.
								\item{\textbf{Join node:}} $v$ has exactly two children $u,w$ in $T$
										and $X_v = X_w = X_u$.
						\end{description}
		\end{itemize}
\end{definition}
\begin{lemma}[\cite{bodlaender2016c} and {\cite[Lemma 7.4]{cygan2015parameterized}}]
		\label{lem:compute-ntd}
Given graph $G$ of treewidth $k$,
one can compute in $2^{\bigO(k)} \cdot n^{\bigO(1)}$ time a nice tree decomposition $\treedec{}$ for $G$ of width $\bigO(k)$ 
such that $\abs{V(T)} \in \bigO(|V(G|)$.
\end{lemma}
\begin{proof}[Proof of \cref{lem:tw-algo}]
Let $I=(G,s,t,c,n_1,n_2,n_3,n_4)$ be an instance of \coloredPath.
Let $n := |V(G)|$ and $k$ be the treewidth of $G$.
By \cref{lem:compute-ntd}, we compute a nice tree decomposition $\treedec$ for $G$ of  width $\bigO(k)$ 
such that $\abs{V(T)} \in \bigO(n)$.
As a first step we add $\{s,t\}$ to every bag of $\TT$.
Henceforth, we say a node $v \in V(T)$ is an introduce/forget/join node if it was such a node before we added $\{s,t\}$ to every bag.
Let $T_v$ be the subtree of $T$ rooted at $v$ and $X(T_v) := \bigcup_{u \in V(T_v)} X_u$.
We are going to compute a dynamic program $D$ such that the following assumption is true for all $v \in V(G)$.
\begin{assumption}
		\label{ass:dp}
Let $v \in V(T)$ be fixed.
For every set of vertex pairs $\Lambda_1, \Lambda_2, \dots, \Lambda_q \subseteq X_v$
and every vector $\gamma \in [n]^4$,
assume that $D_v[\{\Lambda_i\}_{i=1}^q, \gamma] = \top$
if and only if there are internally vertex-disjoint paths $P_1,\dots,P_q$ in $G[X(T_v)]$ 
with $X_v \cap V(P_i) = \Lambda_i$
and $\abs*{\bigcup_{i=1}^q V(P_i) \cap c^p} = \gamma_p$
for all $p \in \oneto{4}$.
\end{assumption}
Hence, we would like to know whether $D_r[\{\{s,t\}\},(n_1,n_2,n_3,n_4)] = \top$.
Note that we could store instead of $\top$ a set of vertex disjoint paths $P_1,\dots,P_q$ as specified in \cref{ass:dp} to compute an actual solution,
but we omit this for the sake of simplicity.
Thus, it is only left to show that we can compute~$D$ in  $k^{\bigO(k)}n^{\bigO(1)}$~time
such that \cref{ass:dp} holds for all $v \in V(T)$.
We do this by induction over the structure of the tree $T$.
\newcommand{\ev}[1]{\vec{e}_{#1}}%
In the following, $\ev{i}$ denotes the $i$-th canonical unit vector.

\subparagraph{Leaf node.} Let $v$ be a leaf in $T$,
hence $X_v = \{ s,t \}$.
We set
$
	D_v[\{\{s,t\}\},\ev{c(s)}+\ev{c(t)}] = \top
$
if and only if $\{s,t\} \in E(G)$.
Then it is easy to verify that \cref{ass:dp} holds for $v$ and all values of $D_v$ can be computed in $n^{\bigO(1)}$ time.

\subparagraph{Forget node.}
Let $v$ be a forget node in $T$ and $w$ the child of~$v$.
Hence $X_v = X_w \setminus \{ u \}$ for some $u \in V(G)$.
Assume that \cref{ass:dp} holds for~$w$.
Let $\Lambda_1, \dots, \Lambda_q \subseteq X_v$ be arbitrary vertex pairs and $\gamma \in \oneto{n}^4$.
We set $D_v[\{\Lambda_i\}_{i=1}^q, \gamma] = \top$
if and only if any of the following two conditions holds.
\begin{enumerate}[(i)]
\item \label{forget:case1} $D_w[\{\Lambda_i\}_{i=1}^q, \gamma] = \top$.
\item \label{forget:case2} There is an index $j \in \oneto{q}$ such that $D_w[\Psi_j(\{\Lambda_i\}_{i=1}^q), \gamma] = \top$
where $\Psi_j$ is the operation of replacing $\Lambda_j =: \{s_j, t_j\}$ by the two pairs $\{s_j, u\}, \{u, t_j\}$.
\end{enumerate}
Note that $X(T_v) = X(T_w)$ and $X_v = X_w \setminus \{ u \}$.
Assume that paths $P_1, \dots, P_q$ as required by \cref{ass:dp} exist for~$v$.
Either none of them contains~$u$, in which case (\ref{forget:case1}) must be true,
or one of them (say $P_j$) does, in which case (\ref{forget:case2}) holds.
Conversely, if one of these conditions holds, then such paths do exist.
Clearly, all values of $D_v$ can be computed in $k^{\bigO(k)}n^{\bigO(1)}$ time.

\subparagraph{Join node.}
Let $v$ be a join node in $T$ and $u, w$ the children of~$v$.
Hence $X_v=X_u = X_w$.
Assume that \cref{ass:dp} holds for $w$ and $u$.
Let $\Lambda_1, \dots, \Lambda_q \subseteq X_v$ be vertex pairs and $\gamma \in \oneto{n}^4$.
We set $D_v[\{\Lambda_i\}_{i=1}^q, \gamma] = \top$ if and only if
$\{\Lambda_i\}_{i=1}^q$
can be partitioned into two sets
$\{\Lambda_{j_{u,i}}\}_{i=1}^{q_u}$,
$\{\Lambda_{j_{w,i}}\}_{i=1}^{q_w}$
and $\gamma$ can be written as a sum $\gamma = \gamma_u + \gamma_w$
such that
\begin{align*}
D_u[\{\Lambda_{j_{u,i}}\}_{i=1}^{q_u}, \gamma_u] &= \top \text{ and}
\\ D_w[\{\Lambda_{j_{w,i}}\}_{i=1}^{q_w}, \gamma_w] &= \top
\end{align*}
Since $X(T_v) = X(T_u) \cup X(T_w)$ and no two vertices of $(X(T_u) \cup X(T_w)) \setminus X_v$ are connected by an edge,
any path in $G[X(T_v)]$ without internal vertices from $X_v$ is either a path in $G[X(T_u)]$ or in $G[X(T_w)]$.
This proves the correctness of this step.
Again, it is easy to see that all values of $D_v$ can be computed in $k^{\bigO(k)}n^{\bigO(1)}$ time.

\subparagraph{Introduce node.}
Let $v$ be an introduce node in $T$ and $w$ be the child of~$v$.
Hence, $X_v=X_w \cup \{ u \}$ for some $u \not\in X_w$.
Assume that \cref{ass:dp} holds for~$w$.
Let $\Lambda_1, \dots, \Lambda_q \subseteq X_v$ be vertex pairs and $\gamma \in \oneto{n}^4$.
Let $J := \{j \in \oneto{q} \mid u \in \Lambda_j\}$
and $\gamma' := \gamma - \sum_{j \in J}\sum_{x \in \Lambda_j} \ev{c(x)}$.
We set $D_v[\{\Lambda_i\}_{i=1}^q, \gamma] = \top$ if and only if
$D_w[\{\Lambda_i\}_{i\in \oneto{q} \setminus J}, \gamma'] = \top$
and each pair $\Lambda_j$ with $j \in J$ forms an edge of~$G$.
Since $X(T_v) = X(T_u) \cup \{u\}$ and $u$~does not have an edge to any vertex of $X(T_u) \setminus X_v$,
any pair~$\Lambda_j$ containing~$u$ must be connected by a direct edge,
while all other pairs must be connected by paths in $G[X(T_u)]$.
From this, the correctness of the above follows.
All values of $D_v$ can clearly be computed in $k^{\bigO(k)}n^{\bigO(1)}$ time.

Having now proven the inductive step for all types of nodes, we conclude that \cref{ass:dp} is true for the root node $r$.
Since $V(T) \in \bigO(n)$,  we have an overall running time of $k^{\bigO(k)}n^{\bigO(1)}$.
This completes the proof of \cref{lem:tw-algo}.
\end{proof}
}%
We are now ready to prove \cref{thm:diverse-paths}.

\begin{proof}[Proof of \cref{thm:diverse-paths}.]
		Let the instance $J$ of \dmPath{} be given in the form of graphs $G_1, \dots, G_\tau$, two vertices $s, t \in \bigcap_{i=1}^\tau G_i$ and $\ell \in \NN$.
	We may assume that every vertex~$v$ of every graph $G_i$ is contained in at least one $s$-$t$ path in $G_i$,
	since otherwise we may delete~$v$.
	This is equivalent to the assumption that the graph $G_i'$ obtained from adding the edge $\{s, t\}$ to $G_i$ is biconnected.
	
	By a result of \citet{MNSW-long-cycles},
	there is a universal constant~$\gamma > 0$, such that
	each~$G_i$ with treewidth $\tw(G_i) \geq \gamma\ell$
	contains two vertex-disjoint cycles of size at least~$4\ell$.
	
	If two such cycles $C, C'$ exist in~$G_i$,
	then let $P_1$ be an $s$-$t$ path containing at least one edge of~$C$.
	To see that such a path exists,
	construct a biconnected graph by simply attaching a new degree-two vertex~$s'$ to both $s$ and $t$,
	create another new vertex~$t'$ by subdividing some edge of~$C$,
	and take two disjoint paths between $s'$~and~$t'$.
	
	Without loss of generality, $P_1$ enters $C$~and~$C'$ at most once each.
	Construct another $s$-$t$ path~$P_2$ from~$P_1$ by setting $E(P_2) := E(P_1) \Delta E(C)$.
	If $P_1$ contains any edge of~$C'$, then define~$P_3$ by $E(P_3) := E(P_1) \Delta E(C')$.
	Otherwise, let $P_3$ be any $s$-$t$ path containing at least half of the edges of~$C'$
	(this can be achieved analogously to the construction of $P_1$ resp.~$P_2$).
	Observe that $P_1, P_2$, and $P_3$ have pairwise symmetric differences at least~$2\ell$.
	Thus, $\{P_1, P_2, P_3\}$ is an $\ell$-diverse representative of all $s$-$t$ paths in $G_i$ by \cref{lem:three-matchings}.
	
	We can then solve the subinstances given by $(G_j)_{j<i}$ and $(G_j)_{j > i}$ separately and 
	pick a suitable path from $\{P_1, P_2, P_3\}$ afterwards.
	
	All subinstances in which every graph~$G_i$ has $\tw(G_i) < \gamma \ell$
	can be solved by \cref{thm:diverse-multistage-pi-det} in combination with \cref{lem:tw-algo}
	in $2^{\bigO(\ell)} f(\gamma \ell) |J|^{\bigO(1)}$~time,
	where $f$~is given by \cref{lem:tw-algo}.
\end{proof}

\toappendix{
\section{Application: Spanning Forests and Other Matroids}
\newcommand{\Matroid}{\textsc{Matroid Independent Set}}
\newcommand{\dmMatroid}{\textsc{Diverse Multistage \Matroid}}
\newcommand{\coloredMatroid}{\textsc{$s$-Colored Exact \Matroid}}
\label{sec:matroids}
In this section, we apply our framework in the context of matroid theory which abstracts the notion of linear independence in vector spaces
and finds applications in geometry, topology, combinatorial optimization, network theory, and coding theory \cite{neel2009matroids,kashyap2009applications}. 
In the classical non-diverse multistage setting, matroids have already been studied \cite{gupta2014changing}. 
We first introduce some standard notations 
and then define the problem we apply our framework to.

A pair $(U,\II)$, where $U$~is the \emph{ground set}
and $\II\subseteq 2^U$ is a family of \emph{independent sets}\footnote{Note that this is \emph{not} a set of pairwise non-adjacent vertices in a graph.},
is a \emph{matroid} if the following holds:
\begin{itemize}
	\item  $\emptyset \in \II$.
	\item  If $A' \subseteq A$ and $A \in \II$, then $A' \in \II$.
	\item  If $A,B \in \II$ and $|A| < |B|$, then there is an $x \in B \setminus A$ such that $A \cup \{x\} \in \II$.
\end{itemize}
Throughout this section, we assume to have access to an oracle which tells us in polynomial time whether a given set $A \subseteq U$ is an element of $\II$.
An inclusion-wise maximal independent set~$A\in \II$
of a matroid~$M=(U,\II)$ is a \emph{basis}.
The cardinality of the bases of~$M$ is 
called the \emph{rank} of~$M$.
A \emph{cycle matroid} of an undirected graph $G$ is a matroid~$(E(G),\II)$, 
where an $A \subseteq E(G)$ is in $\II$ if and only if $(V,A)$ is a forest.
A \emph{partition matroid} is a matroid~$(U,\II)$
such that $\{\II:=\{S\subseteq U\mid | U_i \cap S|\leq r_i, i \in [m]\}$,
where $\bigcup_{i=1}^m U_i$ is a partition of $U$ and $r_i \in \NN_0,i \in [m],m \in \NN$.
A partition matroid is called \emph{uniform matroid} if $m=1$.
Later we will use partition matroids 
to encode constraints
of type ``at most $r_i$~elements of color~$i$''.

We now can formulate the central problem of this section.
\problemdef{\dmMatroid}
{A sequence of matroids $(M_i=(U_i,\II_i))_{i=1}^\tau$ 
		with $\lifetime$ many weight functions $\omega_i \colon U_i \to \NN_0$,
and integers $x_i, \ell \in \NN_0$, for all $i \in [\lifetime]$.}
 {Is there a sequence $(S_1,S_2,\dots,S_\lifetime)$ such that
		 for all $i \in [\lifetime]$ the set $S_i$ is in $\II_i$ 
		 and is of weight $w_i(S_i) \geq x_i$, 
	and for all $i \in [\lifetime-1]$ it holds true that $\abs{S_i \Delta S_{i+1}} \geq \ell$?}
	Note that \dmVoting{} is a special case of \dmMatroid{}, 
	where matroid $M_i$ is a uniform matroid of rank~$k$ with the set of candidates as ground set and 
	the weight functions map to the number of agents approving the candidate.
	We show fixed-parameter tractability \dmMatroid{} parameterized by $\ell$.
\begin{theorem}
		\label{thm:dm-matroid}	
		\dmMatroid{} parameterized by $\ell$ can be solved in $2^{\bigO(\ell)}\cdot n^{O(1)}$ time, where $n$ is the input size.
\end{theorem}
In the context of our framework, $\Pi$ here is the following: 
given an instance $I=(M,w,x)$ consisting of a matroid $M=(U,\II)$,
a weight function $\omega \colon U \to \NN_0$,
and an integer $x$,
is $\solutionset = \{ S \in \II \mid \omega(S)\geq x \}$ non-empty, where $\BB(I)= U$?
Hence, to apply \cref{thm:diverse-multistage-pi}, we study the following problem.
\optproblemdef{\coloredMatroid}
{A matroid $M=(U,\II)$, 
		a weight function voting profile $\omega \colon U \to \NN$,
 a coloring $c\colon U \to [s]$, and
 integers $n_i,x,\in \mathbb N,i\in[s]$.}
 {(if exists) An independent set $S \in \II$ 
 such that $\omega(S) \geq x$ and $|c^{-1}(i) \cap S|=n_i,i \in [s]$?}
 We show that \coloredMatroid{} can be solved efficiently. 
 \begin{lemma}
		\label{lem:coloredMatroid}
		 \coloredMatroid{} is polynomial time solvable.
 \end{lemma}
 \begin{proof}
		 Given an instance $I=(M=(U,\II),\omega,c,n_1,n_2,\dots,n_s,x,k)$ of \coloredMatroid{}.
		Partition the ground set $U =: \biguplus_{i=1}^s U_i$,
		where $U_i := \{ v \in U \mid c(v)=i \}$, for all $i \in [s]$.
		Construct the partition matroid $M' := (U, \II' := \{ S \subseteq U \mid |S \cap U_i| \leq n_i \})$.
		Now compute in polynomial time a set $S \in \II \cap \II'$ of size $\sum_{i=1}^4 n_i$ of maximum-weight with respect to $\omega$ \cite[Section~4.1 and 41.3a]{Sch03}. 
		Thus, we clearly have $S \in \II$.
		Observe that $S \in \II'$ implies $|c^{-1}(i) \cap S| \leq n_i,i \in [4]$.
		Since the set $S$ is of size $\sum_{i=1}^4 n_i$ ,
		we have $|c^{-1}(i) \cap S| = n_i,i \in [s]$.
		Hence, if $w(S) \geq x$, 
		then we correctly output $S$.
		Otherwise, we terminate without an output.
		Observe that this is correct since we have for 
		any set $S \in \II$ with $|c^{-1}(i) \cap S| = n_i,i \in [s]$ that~$S \in \II \cap \II'$.
 \end{proof}
 Now \cref{thm:dm-matroid} follows from \cref{thm:diverse-multistage-pi-det,lem:coloredMatroid}.
 Among others, \cref{thm:dm-matroid} implies that the following problem is fixed-parameter tractable when parameterized by $\ell$.
 \problemdef{\textsc{Diverse Multistage Spanning Forests}}
{A sequence of %
graphs $(G_i)_{i=1}^\lifetime$ and $\ell \in  \NN_0$.}
{Is there a sequence $(S_1,S_2,\dots,S_\lifetime)$ such that
		for all $i \in [\lifetime]$ the graph $(V(G_i),S_i)$ is a spanning forest of $G_i$,
	and for all $i \in [\lifetime-1]$ it holds true that $\abs{S_i \Delta S_{i+1}} \geq \ell$?}
Spanning forests have been studied by \citet{gupta2014changing} in the non-diverse multistage setting.
}%

\section{Hardness of Vertex Cover}
\label{sec:vertex-cover}
\appendixsection{sec:vertex-cover}

\newcommand{\dmVC}{\textsc{Diverse Multistage Vertex Cover}}
We finally present a problem where our framework from \cref{sec:framework} is not applicable, unless $\FPT=\Wone$.
The non-diverse variant of the following problem was studied by \citet{FluschnikNRZ19}. Among others, they showed \Wone-hardness 
when parametrized by the vertex cover size~$k$ or by the maximum number of edges over all instances in the input.
\problemdef{\dmVC}
{A sequence of %
graphs~$(G_i)_{i=1}^{\tau}$ and~$k, \ell \in \NN$.}
{
	Is there a sequence $(S_1,S_2,\dots,S_\lifetime)$ such that
	for all~$i \in \oneto{\lifetime}$ the set~$S_i \subseteq V(G_i)$ is a vertex cover of size at most~$k$ in~$G_i$
	and~$\abs{S_i \Delta S_{i+1}} \ge \ell$ for all~$i \in \oneto{\tau-1}$?
}
The framework from \cref{sec:framework} is presumably not applicable to \dmVC{} because of the following result.
\ifarxiv{}
\begin{theorem}
\else{}
\begin{theorem}
		[\appref{thm:dmvc-hardness}]
\fi{}
	\label{thm:dmvc-hardness}
	\dmVC{} parameterized by $\ell$ is \Wone{}-hard, even if~$\tau=2$.
\end{theorem}
\appendixproof{thm:dmvc-hardness}{
\begin{proof}
	We reduce from \textsc{Independent Set}:
	Given an %
	graph~$G = (V, E)$, and~$k\in\NN$, is there a vertex set~$S \subseteq V$, $|S|\ge k$, such that the vertices in~$S$ are pairwise nonadjacent?
	\textsc{Independent Set} is \Wone{}-hard with respect to~$k$ \cite{downey2013fundamentals}.

	Let~$I := (G=(V,E),k)$ be an instance of \textsc{Independent Set} and let~$|V|=n$.
	\Wilog{}, we assume that~$k>1$.
	We construct an instance~$J := ((G_1, G_2), k', \ell))$ of \dmVC{} as follows.
	The first graph~$G_1$ is a complete graph on the vertex set~$V \cup \{v\}$.
	The second graph~$G_2$ consists of the vertex set~$V \cup \{v\}$ and the edge set~$E \cup \{\{u, v\}\mid u \in V\}$, that is, $G_2$ is a copy of~$G$ to which we add a vertex~$v$ which is adjacent to every other vertex.
	Lastly, we set~$k':=n$ and~$\ell:=k+1$.
	Clearly, $J$ can be constructed in polynomial time.
	We now show that $I$ is a \yes-instance if and only if $J$ is a \yes-instance.

	\textbf{($\Rightarrow$):} 
	Let~$S$ be an independent set of size at least~$k$ in~$G$.
	Let~$S_1 := V$ and~$S_2 := \{v\} \cup V \setminus S$.
	Note that~$|S_1 \Delta S_2| \ge k+1$ and~$|S_2| \le |S_1| = n$,
	and~$S_i$ is a vertex cover in~$G_i$ for~$i \in \oneto{2}$.
	Thus~$(S_1, S_2)$ is a valid solution for our instance of \dmVC{}.

	\textbf{($\Leftarrow$):}
	Let~$(S_1, S_2)$ be a solution for our instance of \dmVC{}.
	As~$G_1$ is a complete graph, we have~$|S_1| \ge n$.
	\Wilog{}, we assume that~$S_1=V$.
	Then~$S_1 \Delta S_2 = \{v\} \cup V \setminus S_2$.
	Note that~$v\in S_2$, otherwise~$S_2$ must be equal to $V$ in order to be a vertex cover, and $|S_1 \Delta V|<\ell$.
	As~$S_2$ is a vertex cover of~$G_2$, the set~$S := (V \cup \{v\}) \setminus S_2$ is an independent set of~$G_2$.
	Note that~$v \notin S$, hence~$S$ is also an independent set of~$G$.
	Finally, as~$S = (S_1 \Delta S_2) \setminus \{v\}$, we have~$|S| \ge \ell-1 = k$, and we are done.
\end{proof}
} %

\section{Conclusion}

We introduced a versatile framework to show fixed-parameter tractability for a variety of diverse multistage problems
when parameterized by the diversity~$\ell$.
The only requirement for applying our framework is that a four-colored variant of the base problem can be solved efficiently.
We presented four applications of our framework, one of which resolving an open question by \citet{bredereck2020multistage}.
Two other applications revealed problems which may be of independent interest from a technical and motivational point of view, see \cref{sec:matching,sec:path}.

We believe that our framework can be applied to a broad spectrum of multistage problems.
In particular, a broad systematic study of the multistage setting in elections was proposed by \citet{BN21}.
Herein, diversity is a natural goal.
From a motivational point of view, an interesting direction for future research is to combine the diverse multistage setting 
with time windows, known from other temporal domains \cite{Z20,MMNZZ,CasteigtsHMZ20,akrida2020temporal,mertzios2019sliding}.
Here, a solution to the $i$-th instance should be sufficiently different from the $\delta$ previous solutions in the sequence; 
our work covers the case~$\delta = 1$.
In some multistage scenarios a ``global view'' \cite{heeger2021multistage} on the symmetric differences is desired.
In context of this paper this means that two consecutive solutions can have a small symmetric difference as long as the sum of all consecutive symmetric differences is at least~$\ell$.
We believe that our framework (\cref{sec:framework}) can be extended to this setting. 
To see this, we have to realize that for an $\ell$-diverse representative $\mathcal F$ of a family of solutions the following holds:
For all sets $A$ and $B$ and integers $\ell_a, \ell_b \leq \ell$, if there is an $S \in \mathcal F$ 
such that $\abs{A \Delta S} \geq \ell_a$ and
$\abs{B \Delta S} \geq \ell_b$, 
then there is an $\widehat{S} \in \mathcal F$
such that $\abs{A \Delta \widehat{S}} \geq \ell_a$ and
$\abs{B \Delta \widehat{S}} \geq \ell_b$.
We leave the details for further research.
Finally, the presented time and space constraints to compute $\ell$-diverse representatives seem to be suboptimal.
Hence, improving the time or space constraints could be a fruitful research direction.

\bibliography{strings-long,bibliography}

\ifarxiv{}\else{}
\appendix
\appendixProofText
\fi{}

\end{document}